\documentclass[11pt]{article}

\PassOptionsToPackage{numbers, compress}{natbib}
\usepackage{amsmath,amssymb,amsthm,fullpage,mathrsfs,pgf,tikz,caption,subcaption,mathtools,mathabx}
\usepackage[ruled,vlined,noresetcount]{algorithm2e}
\usepackage{natbib}
\usepackage{amsmath,amssymb,amsthm,mathtools}
\usepackage[utf8]{inputenc} 
\usepackage[T1]{fontenc}    
\usepackage{hyperref}       
\usepackage{url}            
\usepackage{booktabs}       
\usepackage{amsfonts}       
\usepackage{nicefrac}       
\usepackage{microtype}      
\usepackage{times}
\usepackage{bbm}
\usepackage{enumitem}
\usepackage{xcolor}
\usepackage{cleveref}
\usepackage{bm}
\newtheorem{theorem}{Theorem}[section]
\newtheorem{corollary}[theorem]{Corollary}

\newtheorem{lemma}[theorem]{Lemma}
\newtheorem{definition}[theorem]{Definition}

\newtheorem{remark}[theorem]{Remark}
\newtheorem{observation}[theorem]{Observation}
\newtheorem{claim}[theorem]{Claim}

\newcommand{\R}{\mathbb{R}}

\newcommand{\FF}{\mathbb{F}}

\DeclareMathOperator{\Ima}{im}

\newcommand{\floor}[1]{\lfloor {#1} \rfloor}
\newcommand{\ceil}[1]{\lceil {#1} \rceil}
\newcommand\norm[1]{\left\lVert#1\right\rVert}
\newcommand\hnorm[1]{\left|#1\right|}

\newcommand{\eps}{\varepsilon}

\newcommand{\cG}{\mathcal{G}}
\newcommand{\cC}{\mathcal{C}}
\newcommand{\nonl}{\renewcommand{\nl}{\let\nl\oldnl}}

\usepackage[margin=1in]{geometry}
\author{Max Hopkins\thanks{Department of Computer Science and Engineering, UCSD, CA 92092. Email: \texttt{nmhopkin@eng.ucsd.edu}. Supported by NSF Award DGE-1650112.} \and Ting-Chun Lin\thanks{Department of Physics, UCSD, CA 92092, and Hong Hai Research Institute, Taipei. Email: \texttt{til022@ucsd.edu}.}}
\title{Explicit Lower Bounds Against $\Omega(n)$-Rounds of Sum-of-Squares}
\begin{document}
\maketitle

\begin{abstract}
We construct an explicit family of 3-XOR instances hard for $\Omega(n)$-levels of the Sum-of-Squares (SoS) semi-definite programming hierarchy. Not only is this the first explicit construction to beat brute force search (beyond low-order improvements (Tulsiani 2021, Pratt 2021)), combined with standard gap amplification techniques it also matches the (optimal) hardness of random instances up to imperfect completeness (Grigoriev TCS 2001, Schoenebeck FOCS 2008). 

Our result is based on a new form of \emph{small-set} high dimensional expansion (SS-HDX) inspired by recent breakthroughs in locally testable and quantum LDPC codes. Adapting the recent framework of Dinur, Filmus, Harsha, and Tulsiani (ITCS 2021) for SoS lower bounds from the Ramanujan complex to this setting, we show any (bounded-degree) SS-HDX can be transformed into a highly unsatisfiable 3-XOR instance that cannot be refuted by $\Omega(n)$-levels of SoS. We then show Leverrier and Z\'emor's (Arxiv 2022) recent qLDPC construction gives the desired explicit family of bounded-degree SS-HDX. Incidentally, this gives the strongest known form of bi-directional high dimensional expansion to date.

\end{abstract}
\newpage 
\section{Introduction}\label{sec:intro}
The Sum-of-Squares (SoS) semi-definite programming (SDP) hierarchy is one of the most powerful and widely studied algorithmic frameworks for approximating constraint satisfaction problems (CSPs) in theoretical computer science, yet very little is known about the structure of instances that are \textit{hard} for the paradigm. Indeed, while it has long been known that \textit{random} instances of CSPs are hard for Sum-of-Squares \cite{grigoriev2001linear,schoenebeck2008linear,tulsiani2009csp,barak2015sum,chan2016approximation,kothari2017sum}, there are essentially no \textit{explicit} constructions of hard instances better than brute force search \cite{dinur2020explicit,Tulsianibrute,PrattCommunication}. Leveraging recent breakthroughs in locally testable \cite{dinur2021,lin2022c} and quantum low-density parity-check (qLDPC) codes \cite{panteleev2021asymptotically, leverrier2022quantum}, we resolve this problem, giving the first explicit family of highly unsatisfiable CSPs that cannot be refuted by $\Omega(n)$-rounds of Sum-of-Squares.
\begin{theorem}[Main Result: Explicit 3-XOR Instances Hard for SoS]\label{intro:thm-main}
There exist constants $\mu_1,\mu_2 \in (0,1)$ and an infinite family of 3-XOR instances constructable in deterministic polynomial time such that:
\begin{enumerate}
    \item No assignment satisfies more than a $1-\mu_1$ fraction of constraints
    \item No instance can be refuted by $\mu_2n$ levels of the corresponding Sum-of-Squares SDP Relaxation.
\end{enumerate}
\end{theorem}
Though \Cref{intro:thm-main} only exhibits an `integrality gap' of $1$ v.s $1-\mu_1$ (meaning the instance are $(1-\mu_1)$-satisfiable but \textit{look} fully satisfiable to SoS), combined with standard PCP-like reductions in the SoS hierarchy this gap can be amplified to $1-\varepsilon$ v.s $\frac{1}{2}+\varepsilon$ for any $\varepsilon>0$ \cite{tulsiani2009csp,dinur2020explicit}, which matches the hardness of random $3$-XOR instances up to imperfect completeness \cite{grigoriev2001linear,schoenebeck2008linear}.\footnote{Indeed one can see such a gap is essentially optimal, as a random assignment to any $3$-XOR instance will satisfy $1/2$ the constraints in expectation.} In fact, it is worth noting that \Cref{intro:thm-main} is the first explicit family of CSPs to even beat more than $O(\log(n))$ levels of the SoS hierarchy, which can be done either by unique neighbor expanders \cite{PrattCommunication, alon2002explicit} or (up to lower order factors) simply by brute force search \cite{Tulsianibrute}. While explicit constructions against $\Omega(n)$-rounds of SoS were known in proof complexity (e.g.\ Tseitin formulas \cite{grigoriev1998tseitin}, knapsack \cite{grigoriev2001complexity}), these examples do not lead to inapproximability since their satisfiability is not bounded away from $1$.

Thus, at a high level, \Cref{intro:thm-main} provides the first example of an approximation problem with short witnesses of unsatisfiability that cannot be captured by the Sum-of-Squares proof system, settling (in the negative) the completeness of SoS in this setting. Furthermore, it is worth noting that $3$-XOR is not somehow `special' in this sense. As observed in \cite{dinur2020explicit} (who showed an analogous result for $O(\sqrt{\log(n)})$-levels of SoS), \Cref{intro:thm-main} also gives explicit hard instances across many types of CSPs by standard reduction techniques \cite{tulsiani2009csp}, including instances with optimal integrality gaps for CSPs with approximation resistant predicates based on pairwise independent subgroups \cite{chan2016approximation,dinur2020explicit}.


\subsection{High Dimensional Small-Set Expanders}
\Cref{intro:thm-main} is based on a new form of \textit{high dimensional expansion} (HDX), a nascent area of computer science and math that has already seen an impressive array of breakthrough results across areas such as coding theory \cite{jeronimo2021near, dinur2021,panteleev2021asymptotically,lin2022c}, approximate sampling \cite{kaufman2020high,anari2019log,alev2020improved,anari2020spectral}, approximation algorithms \cite{alev2019approximating,bafna2022high}, analysis of boolean functions \cite{dikstein2018boolean,bafna2021hypercontractivity,gur2021hypercontractivity}, agreement testing \cite{dinur2017high,dikstein2019agreement}, and, recently, Sum-of-Squares lower bounds \cite{dinur2020explicit}. While most of these works consider notions of expansion on \textit{hypergraphs} (often called \textit{simplicial complexes} in this setting), we take inspiration from recent breakthroughs on LTCs \cite{dinur2021,lin2022c} and quantum codes \cite{panteleev2021asymptotically,leverrier2022quantum} and consider expansion on the more general class of \textit{chain complexes}:
\[
X: \mathbb{F}_2^{X(0)} \overset{\delta_0}{\underset{\partial_1}{\rightleftarrows}} \mathbb{F}_2^{X(1)} \overset{\delta_1}{\underset{\partial_2}{\rightleftarrows}} \mathbb{F}_2^{X(2)}.
\]
Here $X(0)$, $X(1)$, and $X(2)$ are sets, $\delta_0$ and $\delta_1$ are linear maps (called the \textit{co-boundary operators}), $\partial_2$ and $\partial_1$ are their transposes (called the \textit{boundary operators}), and both satisfy $\partial_1\partial_2=0$, $\delta_1\delta_0=0$.

Chain complexes admit a natural analog of boundary (edge) expansion in graphs called high-dimensional (co)-boundary expansion \cite{linial2006homological}. To see this, we first note an important inherent structural property of chain complexes: any function $f \in \Ima(\delta_0)$ (called a \textit{co-boundary}) satisfies $|\delta_1 f|=0$. A complex is called a $\rho$-\textit{co-boundary expander} essentially when this is the only obstruction to $|\delta_1f|$ being large:
\[
\forall f \in \mathbb{F}_2^{X(1)}: |\delta_1 f| \geq \rho \cdot d(f, \Ima(\delta_0)).
\]
For intuition, it is worth briefly discussing why this generalizes boundary expansion on graphs. Any graph $G=(V,E)$ (or indeed hypergraph, see \Cref{sec:chain}) can be written as a chain complex:
\[
X: \mathbb{F}_2^{\emptyset} \overset{\delta_0}{\underset{\partial_1}{\rightleftarrows}} \mathbb{F}_2^{V} \overset{\delta_1}{\underset{\partial_2}{\rightleftarrows}} \mathbb{F}_2^{E},
\]
where $\delta_0f(v) = f(\emptyset)$, $\delta_1f((u,v)) = f(u) \oplus f(v)$, and it is easily checked that $\delta_1\delta_0=0$. Notice that in this setting the only co-boundaries are $\Ima(\delta_0)=\{\emptyset, V\}$, and furthermore that for any $S \subset V$ and $e \in E$, the value of $\delta_11_S$ on $e$ is $1$ iff $e$ crosses the cut defined by $S$. This implies the ratio $\frac{\hnorm{\delta_1 1_S}}{d(1_S,\Ima(\delta_0))} = \frac{E(S,V \setminus S)}{\min\{|S|,|V \setminus S|\}}$, which is just the standard boundary expansion of $G$!

Unfortunately, while standard boundary expansion on (random) graphs has been quite useful for proving SoS lower bounds in the past \cite{ben1999short,grigoriev2001linear, schoenebeck2008linear}, high dimensional co-boundary expansion seems to be too strong a notion for this setting: good (co)-boundary expanders are not known to exist (even probabilistically), and their structure is prohibitively restrictive in other senses as well.\footnote{We'll discuss this issue in \Cref{sec:pf-overview}, but in brief co-boundary expansion implies $\ker(\delta_1) = \Ima(\delta_0)$. Like \cite{dinur2020explicit}, our instances will rely on a function in $\ker(\delta_1) \setminus \Ima(\delta_0)$ to enforce global structure on the CSP that cannot be detected through local algorithms like Sum-of-Squares.} We avoid these issues by introducing a simple relaxation of boundary expansion to \textit{small-sets}:
\begin{definition}[Small-set (Co)-Boundary Expansion]\label{intro:def-SS-boundary}
We call $X$ a $(\rho_1,\rho_2)$-small-set boundary expander if the weight of any `small' function $f \in \mathbb{F}_2^{X(1)}$ satisfying $\hnorm{f} \leq \rho_1|X(1)|$ expands:
\[ 
\hnorm{\partial_1 f} \geq \rho_2\cdot d(f,\Ima(\partial_2)).
\]
Similarly, $X$ is a $(\rho_1,\rho_2)$-small-set co-boundary expander if all $f \in \mathbb{F}_2^{X(1)}$ s.t.\ $\hnorm{f} \leq \rho_1|X(1)|$ satisfy:
\[
\hnorm{\delta_1 f}\geq \rho_2\cdot d(f,\Ima(\delta_0)).
\]
We call $X$ a $(\rho_1,\rho_2)$-small-set HDX (SS-HDX) if it satisfies both the above conditions.
\end{definition}
Small-set (co)-boundary expansion is a direct generalization of small-set expansion on graphs, a notion that lies at the heart of many problems in hardness of approximation (especially with respect to Khot's unique games conjecture \cite{khot2002power,raghavendra2012reductions,subhash2018pseudorandom}). In the next section, we will show how SS-HDX naturally lead to hard instances of XOR for Sum-of-Squares (largely following a similar result of Dinur, Filmus, Harsha, and Tulsiani \cite{dinur2020explicit} for the LSV complex \cite{lubotzky2005explicit}), giving the first connection between hardness of approximation and \textit{high dimensional} small-set expanders.

With this in mind, \Cref{intro:thm-main} boils down to constructing an infinite family of SS-HDX on a growing number of vertices, each of which can be constructed in deterministic polynomial time. While this may seem hopelessly strong, a weaker variant of these requirements was very recently achieved in breakthrough constructions of qLPDC codes by \cite{panteleev2021asymptotically,leverrier2022quantum}. Indeed, it turns out these known constructions are already enough: we show Leverrier and Z\'emor's \cite{leverrier2022quantum} recent qLDPC codes are in fact small-set HDX as well.
\begin{theorem}[Small-Set HDX Exist (informal \Cref{thm:HDX})]\label{intro:thm-expanders}
There exist constants $\rho_1,\rho_2 \in (0,1)$ and an explicit (polynomial time constructable) infinite family of bounded-degree\footnote{A complex is bounded degree roughly if each element in $X(i)$ only has constantly many neighbors with respect to the boundary and co-boundary operators. See \Cref{sec:prelims} for an exact definition.} ($3$-term) chain complexes $\{X_i\}$ satisfying:
\begin{enumerate}
    \item $X_i$ has non-trivial `co-homology,' i.e.\ $\Ima(\delta_0) \neq \ker(\delta_1)$
    \item $X_i$ is a $(\rho_1,\rho_2)$-SS HDX.
\end{enumerate}
\end{theorem}
The guarantees of \Cref{intro:thm-expanders} are stronger than those originally proved by Leverrier and Z\'emor \cite{leverrier2022quantum} (see \Cref{sec:related-work} for discussion), and give the strongest known form of bi-directional high dimensional expansion to date.\footnote{In fact it's worth mentioning we actually prove a stronger guarantee regarding \textit{local} functions. See \Cref{remark:local} and discussion in \Cref{sec:related-work}.} Indeed the expansion is so strong that if one could remove the small-set requirement\footnote{Though it is worth noting one must be careful that the dimension of the cohomology stays large, which requires weakening the expansion guarantee to a related notion called (co)-systolic expansion (the correct notion for qLTC regardless) \cite{eldar2017local}.} or prove similar bounds for a $5$-term chain complex, it would resolve the qLTC conjecture \cite{kaufman2014ramanujan,eldar2017local,lin2022c}, a major open problem in quantum computation.

\section{Proof Overview}\label{sec:pf-overview}
We now overview the constructions and proof techniques underlying our main result (\Cref{intro:thm-main}). Broadly speaking, this breaks into two main steps:
\begin{enumerate}
    \item Show any SS-HDX implies a hard instance of $3$-XOR
    \item Construct an explicit infinite family of SS-HDX.
\end{enumerate}
To start, it will be useful to cover some basic background on CSPs, Sum-of-Squares, and chain complexes in a bit more detail. A more formal treatment is given in \Cref{sec:prelims} and \Cref{sec:prelims2}.

\subsection{Background}
In this work, we study the limitations of the Sum-of-Squares proof system for refuting MAX-$k$-XOR, a widely studied class of constraint satisfaction problems (CSPs). An instance of MAX-$k$-XOR $\mathcal{I}$ consists of a set of variables $\{x_i\}_{i \in [n]}$ and constraints $\{C_i\}_{i \in [m]}$, where each $C_i$ is a boolean function of the form:
\[
C_i(x) = \mathbf{1}\left\{ x_{i_1} \oplus \ldots \oplus x_{i_j} = b_i \right \},
\]
where $j=j(i) \leq k$ and $\{i_1,\ldots,i_j\} \subset [n]$. If all constraints have exactly $k$ variables, we say $\mathcal{I}$ is an instance of $k$-XOR. We will usually omit the indicator $\mathbf{1}$ from notation when clear from context. The \textit{value} of $\mathcal{I}$ is the maximum fraction of constraints that can be satisfied by any assignment, and we say $\mathcal{I}$ is \textit{$(1-\mu)$-satisfiable} if there exists an assignment satisfying at least a $(1-\mu)$ fraction of constraints. We call an infinite family of instances $\{\mathcal{I}_i\}$ \textit{explicit} if each instance can be constructed in deterministic polynomial time in the number of variables.

The Sum-of-Squares semi-definite programming hierarchy is a powerful algorithmic framework for approximating the value of any CSP (or more generally for solving constrained polynomial optimization problems). The hierarchy consists of \textit{rounds} or \textit{levels} of progressively stronger SDP relaxations (see \Cref{alg:SoS}). For the moment, it is enough to know that the round-$t$ SoS relaxation is local\footnote{We note the relaxation does have (low-degree) global consistency checks, so it is not fully a local algorithm in this sense.} in the sense that it ranges over subsets of variables of size at most $t$. We will cover more details on the SoS framework as they arise.

Finally, it will be useful to have some basic terminology corresponding to chain complexes. Recall that a chain complex is a sequence
$X: \mathbb{F}_2^{X(0)} \overset{\delta_0}{\underset{\partial_1}{\rightleftarrows}} \mathbb{F}_2^{X(1)} \overset{\delta_1}{\underset{\partial_2}{\rightleftarrows}} \mathbb{F}_2^{X(2)}$
such that $\partial_1\partial_2=0,~\delta_1\delta_0=0$. Functions in the image of $\partial_2$ and $\delta_0$ are called \textit{boundaries} and \textit{co-boundaries} respectively, and are denoted:
\[
\Ima(\partial_2) = B_1, \ \ \Ima(\delta_0) = B^1.
\]
Functions in the kernel of $\partial_1$ and $\delta_1$ are called \textit{cycles} and \textit{co-cycles} respectively, and are denoted:
\[
\ker(\partial_1) = Z_1, \ \ \ker(\delta_1) = Z^1.
\]
The structure of a chain complex promises that $B_1 \subset Z_1 \subset  \mathbb{F}_2^{X(1)}$ and $B^1 \subset Z^1 \subset  \mathbb{F}_2^{X(1)}$. This leads to notions of \textit{homology} and \textit{co-homology} given by (co)-cycles mod (co)-boundary and respectively denoted:
\[
H_1 = Z_1/B_1, \ \ H^1 = Z^1/B^1,
\]
where $G/H$ denotes the quotient group. A complex has non-trivial co-homology if $B^1 \neq Z^1$.
\subsection{From SS-HDX to Hardness}
With notation out of the way, we can now discuss how to transform an expanding chain complex into a hard instance of $3$-XOR. Before we give an informal theorem statement to this effect, it is instructive to overview how one even relates a CSP to a chain complex at all. To this end, let's first recall the classical construction of CSPs (also frequently seen in coding theory) based upon a bipartite graph $B=(L,R,E)$. In this setting, elements in $L$ correspond to variables $\{x_v\}_{v \in L}$, and elements in $R$ correspond to the set of constraints $\{C_r\}_{r \in R}$. Fixing some assignment $\beta \in \{0,1\}^R$ to constraints, the XOR instance classically associated with the graph $B$ is characterized by ensuring the (mod $2$) sum across neighbors of each $r \in R$ is given by $\beta(r)$:
\begin{equation}\label{eq:intro-XOR}
C_r \coloneqq \left\{\sum\limits_{v \in N(r)} x_v = \beta(r) \quad (\text{mod $2$})\right\}.
\end{equation}
In prior hardness constructions, $B$ is typically picked at random in order to satisfy strong expansion properties, while $\beta$ is typically chosen at random to ensure un-satisfiability (see e.g. \cite{grigoriev2001linear,schoenebeck2008linear,kothari2017sum}). While it is sometimes possible to de-randomize the choice of $B$ and retain good inapproximability guarantees, 
no de-randomization of $\beta$ better than brute force search over $\log(n)$-size instances was known up until this point.


The basic form of our XOR instances from chain complexes is actually very similar to \Cref{eq:intro-XOR} (indeed they can be viewed as a special instantiation of this framework). Recall that a chain complex is a sequence:
\[
X: \mathbb{F}_2^{X(0)} \overset{\delta_0}{\underset{\partial_1}{\rightleftarrows}} \mathbb{F}_2^{X(1)} \overset{\delta_1}{\underset{\partial_2}{\rightleftarrows}} \mathbb{F}_2^{X(2)},
\]
and in particular that the co-boundary operator $\delta_0: \mathbb{F}_2^{X(0)} \to \mathbb{F}_2^{X(1)}$ is a linear map. To define an instance of XOR on $X$, we simply move to the \textit{graph representation} of $\delta_0$. Namely, recall that any linear operator mapping from $\mathbb{F}_2^{X(0)}$ to $\mathbb{F}_2^{X(1)}$ can be written as an $(|X(1)| \times |X(0)|)$-dimensional matrix over $\mathbb{F}_2$. We can think of this matrix as the bipartite adjacency matrix of a graph on left vertex set $L=X(0)$ and right vertex set $R=X(1)$. Thus given a function $\beta \in \mathbb{F}_2^{X(1)}$, we construct the associated XOR instance, denoted $\mathcal{I}_{X,\beta}$ as in \Cref{eq:intro-XOR} by adding the constraint for each $r \in X(1)$:
\begin{equation}\label{eq:chain-csp}
C_r \coloneqq \left\{\sum\limits_{\underset{e_r^T\delta_0 e_v = 1}{v \in X(0):} } x_v = \beta(r) \quad (\text{mod $2$})\right\},
\end{equation}
where $e_v \in \mathbb{F}_2^{X(0)}$ and $e_r \in \mathbb{F}_2^{X(1)}$ are the standard basis vectors associated to $v \in X(0)$ and $r \in X(1)$. Note that $e_r^T\delta_0 \in \mathbb{F}_2^{X(0)}$ is just the list of neighbors of $r$, so this is indeed an instantiation of the standard bipartite framework. We note that this construction also generalizes the recent approach of \cite{dinur2020explicit} who built XOR instances via a $3$-dimensional simplicial complex ($4$-uniform hypergraph) by letting \textit{triangles} correspond to constraints, and \textit{edges} correspond to variables. This is exactly the result of the above construction when applied to the natural chain complex associated with a $3$-dimensional simplicial complex (see \Cref{sec:chain} for further details).

So far, we have not used the fact that $\delta_0$ is part of a chain complex, or even the fact that the higher dimensional component $X(2)$ exists at all. This structure comes into play in the choice of $\beta$. Notice that by construction, the instance corresponding to $X$ and a choice of $\beta$ is satisfiable exactly when $\beta$ is a co-boundary. Following the framework laid out in \cite{dinur2020explicit}, the idea is to choose $\beta \in Z^1\setminus B^1$, a function which is a co-cycle, but not a co-boundary. On a sufficiently expanding complex, this choice induces \textit{global structure} on the XOR instance that cannot be captured by \textit{local views} of the complex, where both the homology and co-homology look trivial. Since Sum-of-Squares only looks over local views in this sense, this leads to the following direct translation between SS-HDX and hard instances of XOR.
\begin{theorem}[SS-HDX $\implies$ Hard XOR Instance (Informal \Cref{thm:HDX-to-kXOR})]\label{pf-overview:HDX-to-kXOR}
Let $X: \mathbb{F}_2^{X(0)} \overset{\delta_0}{\underset{\partial_1}{\rightleftarrows}} \mathbb{F}_2^{X(1)} \overset{\delta_1}{\underset{\partial_2}{\rightleftarrows}} \mathbb{F}_2^{X(2)}$ be an SS-HDX with non-trivial co-homology. Then there exist $\mu_1,\mu_2 \in (0,1)$ such that for any $\beta \in Z^1\setminus B^1$, the associated XOR instance $\mathcal{I}_{X,\beta}$ satisfies:
\begin{enumerate}
    \item Soundness: $\mathcal{I}_{X,\beta}$ is at most $(1-\mu_1)$-satisfiable,
    \item Completeness: $\mathcal{I}_{X,\beta}$ cannot be refuted by $\mu_2|X(0)|$ levels of the SoS hierarchy.
\end{enumerate}
\end{theorem}
Before moving on to the construction of SS-HDX, let's discuss how small-set expansion implies soundness and completeness for these instances. Soundness, the simpler of the two, intuitively comes from the fact that small-set co-boundary expansion promises that any element in $Z^1 \setminus B^1$ must be \textit{far from the co-boundary}.\footnote{It is worth noting that this property, called \textit{co-systolic distance}, is quite well studied. Indeed as we will soon discuss it is exactly the property needed (in both directions) to build good qLDPC codes \cite{panteleev2021asymptotically}, and was also used directly by \cite{dinur2020explicit} to prove soundness of their 3-XOR instances by the same argument stated here.} Recall that by construction, the instance $\mathcal{I}_{X,\beta}$ is satisfiable exactly when $\beta \in \mathbb{F}_2^{X(1)}$ is a co-boundary. Intuitively one might then expect that functions which are far from the co-boundary would therefore be far from satisfiable. Indeed this intuition holds true---it is easy to show this robust version of the statement holds for small-set co-boundary expanders, and therefore that our instances are far from satisfiable as well.

Completeness is somewhat trickier and, unlike soundness, does actually require the full power of small-set boundary expansion. We stated earlier that the completeness of our instances, much like those of \cite{dinur2020explicit}, comes from the fact that the global structure of (co)-homology cannot be detected through local views of the complex. This is formalized by observing that small-set boundary expansion can be equivalently re-stated as the following \textit{isoperimetric inequality} (see \Cref{lemma:SS-LTC}): ``small, minimal\footnote{A function $f \in \mathbb{F}_2^{X(1)}$ is said to be minimal if adding any boundary can only increase its size (Hamming weight).} functions have large boundaries.'' Largely following \cite{dinur2020explicit} (who use a much weaker isoperimetric inequality for the LSV complex due to Gromov \cite{gromov1983filling}), the idea is then to combine this fact with the classical arguments of Ben-Sasson and Wigderson \cite{ben1999short} to show that the width\footnote{The width of a refutation is the largest number of variables appearing in any equation.} of any refutation of $\mathcal{I}_{X,\beta}$ in the $\oplus$-resolution proof system\footnote{In this proof system, one is allowed to combine linear equations (equivalently XOR constraints) $\ell_1=b_1$ and $\ell_2=b_2$ to derive the equation $\ell_1 \oplus \ell_2 = b_1 \oplus b_2$. A refutation is a proof based on this rule deriving a contradiction ($0=1$), which is equivalent in our setting to showing the XOR instance is unsatisfiable.} is large. Since Schoenebeck \cite{schoenebeck2008linear} showed any such bound transfers to a completeness lower bound for Sum-of-Squares, this completes the proof.

In slightly more detail, a refutation in the $\oplus$-resolution system can be viewed as an (in-degree two) DAG where leaves correspond to the original XOR constraints, internal nodes correspond to the XOR of their two parents (as in the $\oplus$-resolution derivation rule), and the root derives the contradiction $0=1$. Recall that each element $s \in X(1)$ corresponds to a constraint in our XOR instance. Following \cite{dinur2020explicit}, the idea is to assign a function in $h_v \in \mathbb{F}_2^{X(1)}$ for each node $v$ in the DAG that tracks which XOR constraints are being used at that node. The \textit{boundary} of this function, $\partial_1h_v \in \mathbb{F}_2^{X(0)}$, is exactly the set of variables appearing in the equation corresponding to node $v$. Thus lower bounding the width of the refutation boils down to finding a node with large boundary. 

This is where small-set boundary expansion (namely the isoperimetric formulation) finally comes into play. In particular, the corresponding inequality states that it is enough to find a node $v$ of `medium' weight:\footnote{We note that weight here is not just the standard Hamming weight, but must take into account distance from the boundary as well. See \Cref{sec:HDX-to-XOR}.} small enough that one can apply the inequality, but large enough to result in a large boundary. This can be done by fairly standard potential arguments (see e.g. \cite{ben1999short,dinur2020explicit}) where one sets of up a potential function tracking this weight throughout the DAG, and argues that the leaves have small potential, the root has large potential, and that potential is sub-additive. This implies the existence of an interior node with medium potential and completes the proof. The details are given in \Cref{sec:HDX-to-XOR}.

Finally, before moving on to overviewing our construction of SS-HDX, we note that except in very special cases (e.g.\ the simplicial complexes considered in \cite{dinur2020explicit}), the CSPs given by \Cref{eq:chain-csp} (and therefore also \Cref{pf-overview:HDX-to-kXOR}) are actually instance of MAX-$k$-XOR, not $3$-XOR, where $k$ is given by the maximum degree of the complex. As it turns out, this is not a significant issue because the SS-HDX we construct in the next section are \textit{bounded degree}, meaning not only that every constraint in the XOR has a constant number of variables, but also that every variable only appears in a constant number of constraints. This observation allows us to move to hard instances of $3$-XOR by standard NP-reduction type arguments within the SoS hierarchy \cite{schoenebeck2008linear,tulsiani2009csp} while only losing constant factors in the soundness and levels of hardness for SoS.

\subsection{Constructing SS-HDX}
Now that we know how to transform an expanding chain complex into a hard instance of $3$-XOR, we turn our attention to the construction of such complexes. Our method relies on recent breakthroughs on LTCs \cite{dinur2021, lin2022c} and quantum LDPC codes \cite{panteleev2021asymptotically,leverrier2022quantum}. As such, we'll split this section into three parts: a review of the connection between quantum LDPC codes and expanding chain complexes, the recent qLDPC construction of Leverrier and Z\'emor \cite{leverrier2022quantum}, and our proof of small-set (co)-boundary expansion.

\subsubsection{Quantum LDPC Codes and Chain Complexes}\label{sec:qldpc-chain}
A classical error correcting code is a method of encoding $k$ classical bits into $n>k$ classical bits such that it is possible to recover the original bit string even if the encoded string becomes corrupted. We will consider \textit{linear codes}, which are defined by a linear operator $M: \mathbb{F}_2^n \to \mathbb{F}_2^{n-k}$ called the parity check matrix,\footnote{We note the parity check matrix is traditionally denoted by `$H$,' but this conflicts with the notation for homology.} where the corresponding code $\cC \coloneqq \ker M$.

Similar to the classical setting, a \textit{quantum} code encodes quantum bits into a larger number of quantum bits, but is resistent to two types of corruption: the $X$-type errors (bit flips) and the $Z$-type errors (phase flips). In this work, we will focus on a popular notion of quantum codes called \textit{CSS-codes} \cite{calderbank1996good, steane1996error}, which come with the benefit of having an entirely classical interpretation. In particular, a length $n$ CSS-code is made up of two classical codes $\cC_0 \coloneqq \ker M_0 \subset \FF_2^n$ and $\cC_1 \coloneqq \ker M_1 \subset \FF_2^n$ such that $\cC_0 ^\perp \subset \cC_1$, or equivalently $M_1 \cdot M_0^T = 0$.\footnote{Here $\cC_0 ^\perp$ denotes the dual code, consisting of all elements orthogonal to $\cC_0$. This code is generated by the transpose of the parity check matrix $M_0^T$.}
The dimension of the code is defined as $k = \dim \cC_0 - \dim \cC_1^\perp$, and its \textit{distance} (which measures how much corruption it can handle) is defined as $d = \min(d_x, d_z)$ where
\begin{equation*}
    d_x = \min_{v \in \cC_0 \setminus \cC_1^\perp} |v|, d_z = \min_{v \in \cC_1 \setminus \cC_0^\perp} |v|
\end{equation*}
and $d_x$ ($d_z$) is called the $X$-distance ($Z$-distance). 
The quantum low-density parity-check (LDPC) conjecture, recently resolved by \cite{panteleev2021asymptotically}, states that there exists a family of quantum CSS codes with linear dimension and distance, $k = \Theta(n)$ and $d = \Theta(n)$, where $M_0$ and $M_1$ have at most some constant number of ones in any row or column (and thus are `low-density' parity check matrices).

Since we are promised by definition that $M_1 \cdot M_0^T = 0$, it is easy to see that any CSS-code induces the following chain complex:
\[
X: \mathbb{F}_2^{m_0} \overset{M_0^T}{\underset{M_0}{\rightleftarrows}} \mathbb{F}_2^{n} \overset{M_1}{\underset{M_1^T}{\rightleftarrows}} \mathbb{F}_2^{m_1},
\]
where $m_i = \dim \left(\Ima (M_i)\right)$. Indeed the same holds in reverse as well, given a chain complex 
\[
X: \mathbb{F}_2^{X(0)} \overset{\delta_0}{\underset{\partial_1}{\rightleftarrows}} \mathbb{F}_2^{X(1)} \overset{\delta_1}{\underset{\partial_2}{\rightleftarrows}} \mathbb{F}_2^{X(2)},
\] one obtains a quantum CSS code by letting $M_0 \coloneqq \partial_1$, and $M_1 \coloneqq \delta_1$. 

In fact, it turns out this equivalence between quantum CSS codes and chain complexes runs deeper: all of the discussed properties (e.g.\ distance, LDPC) have analogs in the homological language we developed in the previous section. The classical codes $\cC_0$ and $\cC_1$, for instance, correspond to the cycles and co-cycles of the chain complex ($\cC_0 = Z_1, \cC_1 = Z^1$), while the dual codes $\cC_0^\perp$ and $\cC_1^\perp$ correspond to the co-boundaries and boundaries ($\cC_0^\perp = B^1, \cC_1^\perp = B_1$). The dimension of the code $k$ corresponds to the dimension of the co-homology ($k = \dim H^1$), and the maximum degree of the complex corresponds to the maximum density of the parity check codes (so the bounded-degree and LDPC conditions are equivalent). Finally, the $X$-distance and $Z$-distance of the code correspond to what is known as the \textit{(co)-systolic distance} of the chain complex, the minimum weight of any (co)-cycle that is not a (co)-boundary:
\begin{align*}
    d_x &= \min_{v \in \cC_0 \setminus \cC_1^\perp} |v| = \min_{v \in Z_1 \setminus B_1} |v|,\\
    d_z &= \min_{v \in \cC_1 \setminus \cC_0^\perp} |v| = \min_{v \in Z^1 \setminus B^1} |v|.
\end{align*}

In \cite{panteleev2021asymptotically} and \cite{leverrier2022quantum}, the authors construct two different explicit families of good quantum LDPC codes. This partially solves our problem since the codes correspond to a family of bounded-degree chain complexes with non-trivial co-homology and linear co-systolic distance (which is enough to imply soundness of our XOR construction).
We will show these complexes in fact satisfy the stronger small-set (co)-boundary expansion condition, which as discussed in the previous section further implies completeness and (up to reduction to $3$-XOR) finishes the proof of \Cref{intro:thm-main}.
\subsubsection{Leverrier and Z\'emor's qLDPC Codes}
Before discussing the proof, we need to overview the original construction of \cite{leverrier2022quantum}. A significantly more detailed description of the construction and its associated components is given in \Cref{sec:prelims2} and \Cref{sec:construction}.

Leverrier and Z\'emor's qLDPC codes are based on a classical object called a \textit{Tanner code} \cite{tanner1981recursive}.  Given an $n_0$-regular graph $\cG = (V, E)$ and a linear code $C$ of length $n_0$, the Tanner code $T(\cG, C) \subset \FF_2^E$ is
\begin{equation*}
    \{c \in \FF_2^E : \forall v \in V, c|_{E(v)} \in C\},
\end{equation*}
where $c|_{E(v)} \in \FF_2^{n_0}$ is the vector formed by the values on the edges incident to $v$. Tanner codes have long been used in coding theory. The main insight of \cite{leverrier2022quantum} was to observe that one can construct a quantum CSS code via two Tanner codes coming from a higher-dimensional object called the \textit{left-right Cayley complex}, recently developed in \cite{dinur2021} to construct c3-LTCs.


The left-right Cayley complex corresponding to a group $G$ and two sets of generators $A = A^{-1}$ and $B = B^{-1}$ consists of a vertex set $V=G$, edges given by (left) Cayley graph $C(G,A)$ and (right) Cayley graph $C(G,B)$, and higher-dimensional `squares' of the form $\{g,ag,gb,agb\}$ for $g \in G, a \in A, b \in B$. More formally, \cite{leverrier2022quantum} consider the \textit{double cover} of this complex where:
\begin{itemize}
    \item The vertices are $V = V_0 \cup V_1$ where $V_0 = G \times \{0\}$ and $V_1 = G \times \{1\}$.
    \item The `$A$-edges' and `$B$-edges' are respectively:
        \begin{equation*}
            E_A = \{\{(g, 0), (ag, 1)\} : g \in G, a \in A\}, E_B = \{\{(g, 0), (gb, 1)\} : g \in G, b \in B\}.
        \end{equation*}
    \item The squares are 
        \begin{equation*}
            F = \{\{(g, 0), (ag, 1), (gb, 1), (agb, 0)\} : g \in G, a \in A, b \in B\}.
        \end{equation*}
\end{itemize}
Notice each square contains exactly two vertices in $V_0$ and two vertices in $V_1$. This allows us to think of each square as an \textit{edge} between two vertices in $V_0$ (or $V_1$) and to define corresponding graphs $\cG_0^\square = (V_0, F)$ and $\cG_1^\square = (V_1, F)$. The local view around each vertex in $(g,i) \in \cG_i^\square$ then corresponds to the squares $\{(g, i), (ag, 1-i), (gb, 1-i), (agb, i)\}$ for $a \in A, b \in B$. Assuming $|A|=|B|=\Delta$ for some constant $\Delta$, we will always think about these local views as square matrices with rows indexed by $A$ and columns indexed by $B$.

Leverrier and Z\'emor \cite{leverrier2022quantum} observed that the Tanner codes associated to these graphs, $\cC_0 = T(\cG^\square_0, C_0^\perp)$ and $\cC_1 = T(\cG^\square_1, C_1^\perp)$, give a quantum CSS code (i.e.\ satisfy $\cC_0 ^\perp \subset \cC_1$) whenever the associated local codes $C_0 = C_A \otimes C_B$ and $C_1 = C_A^\perp \otimes C_B^\perp$ are tensors\footnote{The tensor code $C_A \otimes C_B$ is the set of matrices whose rows are given by elements of $C_B$ and columns are given by elements of $C_A$.} of linear codes $C_A \subseteq \mathbb{F}_2^A$ and $C_B \subseteq \mathbb{F}_2^B$. Furthermore, they showed that whenever $C_A, C_B, C_A^\perp, C_B^\perp$ have linear distance and the codes $C_1^\perp = C_A \otimes \FF_2^B + \FF_2^A \otimes C_B$, and $C_0^\perp = C_A^\perp \otimes \FF_2^B + \FF_2^A \otimes C_B^\perp$ satisfy certain robustness properties (see \Cref{sec:robust-tensor-code}), then the associated quantum code has linear distance. \cite{leverrier2022quantum} complete their construction by showing random base codes $C_A, C_B$ satisfy these properties with high probability. Note that because these base codes are constant size, this final step can be brute-forced to maintain explicitness of the construction.

\subsubsection{Proving Small-Set (Co)-Boundary Expansion}
With \cite{leverrier2022quantum}'s construction in hand, we can now sketch the proof of small-set (co)-boundary expansion. As mentioned previously, all other major requirements (e.g.\ non-trivial homology, bounded-degree) already follow from the fact that the complex corresponds to a good qLDPC code. We will focus here on proving small-set \textit{co}-boundary expansion in particular, but we note that small-set boundary expansion follows the same argument by symmetry of \cite{leverrier2022quantum}'s construction. 

With this in mind, recall that small-set co-boundary expansion can equivalently be phrased as an isoperimetric inequality for small, minimal functions (see \Cref{lemma:SS-LTC}). In particular, to show small-set co-boundary expansion for the chain complex
\[
X: \FF_2^{m_0} \xrightarrow{\delta_0 \coloneqq \cC_0^T} \FF_2^n \xrightarrow{\delta_1 \coloneqq \cC_1} \FF_2^{m_1},
\]
it is enough to show there exist constants $\rho_1, \rho_2 \in (0,1)$ such that any minimal $x \in \FF_2^n$ with weight $|x| \le \rho_1 n$ has large boundary: $|\delta_1 x| \ge \rho_2 |x|$. We proceed by contradiction. Assuming $|\delta_1 x| < \rho_2 |x|$, we will show $x$ is not minimal by finding $y \in B^1$ such that $|x+y| < |x|$.

The proof of this fact largely follows the technique of \cite{leverrier2022quantum} for proving the weaker co-systolic distance property. The main difference is that while \cite{leverrier2022quantum} only consider functions $x \in \FF_2^n$ that are co-cycles, we consider arbitrary functions. In particular, recall that the co-cycles in our construction correspond to codewords in the Tanner code $T(\cG_1^\square,C_1^\perp) $, or equivalently to functions $x \in \FF_2^n$ whose `local view' around each vertex $(g,1) \in V_1$ is given by a codeword of $C_1^\perp$. Since our functions do not a priori have this structure, we will need to track the set of `violations' coming from local views that are not codewords (this essentially corresponds to where $\delta_1 x$ is non-zero). 

To this end, recall $x$ is a bit string indexed by the squares of the double-covered Cayley complex, and let $S \subset V_1$ denote the set of vertices incident to any square in $x$. We partition $S$ into three parts: the violated vertices $S_v$, the normal vertices $S_n$, and the exceptional vertices $S_e$. A vertex is \textit{violated} if the local view of $x$ around the vertex does not form a codeword in $C_1^\perp$. When the local view does form a codeword, if the codeword has weight less than $w \coloneqq \Delta^{3/2-\eps}$ we call it \textit{normal}, and otherwise call it \textit{exceptional}.
This weight-based distinction comes from the robustness condition of the local tensor code. We cover this in detail in \Cref{sec:robust-tensor-code}, but for the moment it is sufficient to think of robustness as a structural condition forcing codewords with weight less than $w$ to be zero outside of a small number of rows and columns.
In particular, this promises that each column (respectively row) in the local view of a normal vertex is at most $O(\Delta^{1/2 - \eps})$ away from a codeword in $C_A$ (respectively $C_B$).


Following \cite{leverrier2022quantum}, our goal is now to find a vertex $v \in V_0$ that shares $\Omega(\Delta)$ columns or rows with $S_n$. As long as $S_e$ and $S_v$ are not too large compared to $S_n$, robustness of the code then implies the local view of $v$ is within $O(\Delta^{3/2+\eps})$ of a codeword $c \in C_A \otimes C_B$, but also has total weight $\Omega(\Delta^2)$.\footnote{We note $C_A$ and $C_B$ can be chosen to have linear distance to ensure this.} This means we can construct a vector $y \in B^1$ by defining $y$ to be $c$ on the local view of $v$ and $0$ everywhere else. Since $x+y$ and $x$ match outside the local view (where $x$ has weight $\Omega(\Delta^2)$ and $x+y$ has weight $O(\Delta^{3/2+\eps})$), this implies $|x+y|<|x|$ as desired.

It therefore remains to find such a vertex $v \in V_0$, which is the main technical component of the proof. Let $T \subset V_0$ be the vertices that share at least one `heavy' column or row with a normal vertex (that is one with many 1s). One can equivalently think of this as an edge between $V_0$ and $V_1$ that is `heavy' in the sense that it is contained in many squares in $x$. The idea is then to show that there are many such heavy edges passing between $S$ and $T$. Using expansion of the underlying graph and our assumption $\hnorm{\delta_1 x} < \rho_2 \hnorm{x}$, one can prove that $T$, $S_e$, and $S_v$ are small compared to $S_n$. This implies that a typical vertex in $T$ has not just one, but $\Omega(\Delta)$ heavy edges to $S_n$, which in turn corresponds to sharing $\Omega(\Delta)$ rows and columns with normal vertices and completes the proof.




\section{Discussion}\label{sec:discussion}
\subsection{Related Work}\label{sec:related-work}
\paragraph{Sum-of-Squares Lower Bounds:}
At a conceptual level, our work fits into a long line of research on the limitations of Sum-of-Squares and related proof systems (e.g.\ Nullstellensatz \cite{beame1996lower,grigoriev1998tseitin}, Polynomial Calculus \cite{clegg1996using,buss2001linear}), and LP/SDP hierarchies (e.g.\ Sherali-Adams \cite{charikar2009integrality, mathieu2009sherali,benabbas2012sdp}, Lovász-Schrijver \cite{alekhnovich2005towards, schoenebeck2007tight,georgiou2010integrality}). 
Most relevant to our setting is the line of work on Sum-of-Squares lower bounds initiated by Grigoriev \cite{grigoriev2001linear} (and later independently Schoenebeck \cite{schoenebeck2008linear}), who used boundary expansion to prove random 3-XOR instances cannot be refuted by $\Omega(n)$ levels of SoS. This lead to a number of works improving integrality gaps for more general classes of random $k$-CSPs \cite{tulsiani2009csp, barak2015sum,chan2016approximation,kothari2017sum} along with a number of other combinatorial optimization problems by reduction \cite{schoenebeck2008linear, tulsiani2009csp}.

In a sense, these prior works on SoS lower bounds for random instances can be viewed as increasingly strong and general formulations of the statement: `Sum-of-Squares fails to capture the probabilistic method.' In contrast, Dinur, Filmus, Harsha, and Tulsiani \cite{dinur2020explicit} recently exhibited the first \textit{explicit} families of CSPs hard for Sum-of-Squares based on an algebraic, highly structured family of objects called Ramanujan (or LSV) complexes \cite{lubotzky2005explicit}, suggesting a new paradigm of hardness for structured instances. Due to the poor systolic expansion of the Ramanujan complex, \cite{dinur2020explicit}'s bounds only hold up to $O(\sqrt{\log(n)})$ rounds of SoS as compared to $\Omega(n)$ levels for random instances. Nevertheless, the authors conjectured it might be possible to use such anti-random objects to fool $\Omega(n)$ levels as well. Our work can be viewed as a confirmation of this general hypothesis: anti-random structure (in particular certain \textit{algebraic} structure) is indeed as hard as random for Sum-of-Squares.

\paragraph{High Dimensional Expansion:}
High dimensional expansion in the form we consider (i.e.\ topological expansion) was originally introduced by Linial and Meshulam \cite{linial2006homological} to study the vanishing of cohomology on random simplicial complexes, and independently by Gromov \cite{gromov2010singularities} to study the topological overlapping principle. While our particular notion of small-set (co)-boundary expansion has not been studied in the literature, a stronger isoperimetric inequality for small, \textit{locally} minimal\footnote{A function is locally minimal if its weight cannot be decreased by adding the image of any standard basis vector $\partial_2(e_v)$. Any minimal function is also locally minimal (and the converse does not in general hold), so this is a strictly stronger notion of expansion than we study.} functions was used by Kaufman, Kazhdan, and Lubotzky \cite{kaufman2014ramanujan} to show the existence of bounded degree co-systolic expanders (another weakening of co-boundary expansion that replaces distance from $B^1$ with distance from $Z^1$), and later in \cite{evra2016bounded,kaufman2018cosystolic,kaufman2021unique} as well. A similar strategy was recently employed by Lin and Hsieh to construct c3-LTCs \cite{lin2022c} and later (conditional) qLDPC codes \cite{lin2022good}. It is worth noting that this stronger condition actually holds for our construction as well (see \Cref{remark:local}).

\paragraph{Quantum Codes and LTCs:}
Quantum LDPC and locally testable codes have long been known to share a close connection with topological notions of high dimensional expansion (see e.g.\ \cite{eldar2017local}). Indeed it was qLDPC constructions based on the Ramanujan complex \cite{evra2020decodable,KaufmanT21quantum} that first broke the $\sqrt{n}$ distance barrier and started the race to good qLDPCs \cite{evra2020decodable,breuckmann2021quantum,panteleev2021quantum,KaufmanT21quantum,hastings2021fiber,panteleev2021asymptotically,jeronimo2021explicit,leverrier2022quantum}. As discussed in \Cref{sec:qldpc-chain}, qLDPC codes satisfy a weaker variant of expansion called (co)-systolic distance, but must do so \textit{in both directions}. This is in strong contrast to typical constructions in the HDX literature which, due to the inherent asymmetry of simplicial complexes, typically have very poor boundary expansion (indeed this is also why we avoid simplicial complexes in this work). Such a guarantee was only recently achieved by Panteleev and Kalachev \cite{panteleev2021asymptotically} using refined products of chain complexes, and very recently simplified through a more geometric lens by Leverrier and Z\'emor \cite{leverrier2022quantum}. Since small-set (co)-boundary expansion is a stronger notion than (co)-systolic distance (see \Cref{sec:SS-boundary}), our analysis provides the strongest form of two-sided topological expansion to date. Further, this stronger form of two-sided expansion also gives some hope for a positive resolution of the famous qLTC conjecture. If, for instance, one can construct a $5$-term chain complex satisfying similar bi-directional small set expansion guarantees, qLTC would follow by the arguments of \cite{kaufman2014ramanujan,eldar2017local,lin2022c}.

\subsection{Further Directions}\label{sec:further}
\paragraph{Improved Integrality Gaps:} We prove the existence of an explicit family of $3$-XOR instances with a constant integrality gap of $1$ v.s $1-\mu$ for $3$-XOR, which falls short of reaching the $1$ v.s $\frac{1}{2}+\varepsilon$ gap exhibited by random instances \cite{grigoriev2001linear,schoenebeck2008linear}. While standard reductions in the SoS hierarchy can improve our gap to arbitrarily close ($1-\varepsilon$ v.s $\frac{1}{2}+\varepsilon$), perfect completeness is lost in the process. The same issue was observed in \cite{dinur2020explicit}'s original explicit construction from the Ramanujan complex. They asked whether it is possible to bypass imperfect completeness by giving a direct construction with co-systolic distance at least $\frac{1}{2}-\varepsilon$. This remains a natural open question in our setting as well---can one directly construct a small-set boundary expander with co-systolic distance $\frac{1}{2}-\varepsilon$? This would lead to a $1$ v.s $\frac{1}{2}+\varepsilon$ gap for MAX-$k$-XOR. Another natural question is whether such a bound can be transferred to $3$-XOR without losing factors in the soundness. Our current reduction loses a factor in $k$, but we have made no attempt to optimize this step (since any constant gap is sufficient to amplify with PCP techniques if one is okay with imperfect completeness).

\paragraph{Hardness Beyond XOR:} 
Many of the best integrality gaps known for combinatorial optimization problems (e.g. maximum independent set, chromatic number) are proved by reduction from $k$-CSPs \cite{tulsiani2009csp}. Unfortunately, such reductions are often \textit{randomized}, so they do not imply explicit hard instances even when combined with our XOR construction. This raises a natural question: can we build \textit{explicit} reductions from $k$-CSPs to classical combinatorial problems such as maximum independent set? Combined with our construction, this could lead to new families of hard instances for many well-studied combinatorial optimization problems. On a related note, it is worth observing that these reductions usually rely on CSPs with better integrality gaps than $k$-XOR. For instance, it is not hard to see that while  random instances of $k$-XOR only exhibit a $1$ v.s $1/2+\varepsilon$ integrality gap, more constrained $k$-CSPs (e.g. constraints of the form $Ax=b$ for some matrix $A \in \mathbb{F}_2^{d \times k}$) can lead to much larger integrality gaps up to $1$ v.s $\frac{2k}{2^k}+\varepsilon$ \cite{tulsiani2009csp}. Can we use high dimensional expanders to recover explicit $k$-CSPs matching these bounds? 

\paragraph{Small-Set HDX and Hardness of Approximation:} Small set expansion plays a fundamental role in hardness of approximation, ranging from use as a computational hardness assumption itself \cite{raghavendra2010graph}, to its pivotal use in the proof of the 2-2 games conjecture \cite{khot2017independent,dinur2018towards,dinur2018non,barak2018small,khot2018small,subhash2018pseudorandom} and recent converse use for algorithms for unique games \cite{bafna2020playing,bafna2022high}. This work gives the first application of \textit{high dimensional} small-set expansion to hardness of approximation, raising the natural question: does this high dimensional variant have a broader role to play in the field as well?

\section{Preliminaries I: SS-HDX to Hardness}\label{sec:prelims}
We now cover the preliminary definitions required to understand our general translation of expanding chain complexes into hard instance of 3-XOR, including basics on Sum-of-Squares, chain complexes, and traditional notions of high dimensional expansion. Background required for the HDX construction itself (e.g.\ on left-right Cayley complexes, robust tensor codes, etc.) is postponed to \Cref{sec:prelims2}.
\subsection{Sum of Squares and Refutations}
The Sum-of-Squares Semidefinite Programming Hierarchy is a powerful method for approximately solving constrained polynomial optimization problems, and is in particular the strongest known algorithmic framework for approximating CSPs. In brief, the SoS heirarchy presents a series of successively stronger SDP relaxations of a problem, where the `round-$t$' relaxation optimizes over $t$-local views and runs in time $n^{O(t)}$. We refer the reader to \cite{BS14,fleming2019semialgebraic} for general information on the SoS hierarchy.

In this work, we focus in particular on the SoS relaxations of MAX-$k$-XOR, the family of CSPs on $n$ variables $\{x_1,\ldots, x_n\}$ and $m$ constraints $\{C_i\}_{i \in [m] }$ of the form:
\[
x_{i_1} \oplus \ldots \oplus x_{i_j}= z_i,
\]
where $z_i \in \{0,1\}$, $\{i_1,\ldots,i_j\} \subset [n]$, and $j=j(i) \leq k$. Let $T_i \subset [n]$ denote the set of variables appearing in the $i$th constraint. Then the round-$t$ SoS SDP relaxation for MAX-$k$-XOR can be written as:

\begin{algorithm}[H]
\SetAlgoLined
\nonl \textbf{Input:} variables $\{v_S\}_{S \in {[n] \choose \leq t}}$\\
\nonl \textbf{Maximize:} $\frac{1}{2} + \frac{1}{2m}\sum\limits_{i=1}^m (-1)^{z_i}\langle v_{T_i}, v_{\emptyset} \rangle$\\
\nonl \textbf{Constraint to:} 
\begin{enumerate}[leftmargin=*]
    \item $\forall S_1 \oplus S_2= S_3 \oplus S_4, |S_i| \leq t: \langle v_{S_1},v_{S_2}\rangle = \langle v_{S_3},v_{S_4} \rangle$
    \item $\forall S, |S|\leq t: \norm{v_S}_2=1$
\end{enumerate}
 \caption{Round-$t$ SoS Relaxation for MAX-$k$-XOR}
 \label{alg:SoS}
\end{algorithm}
We refer to the maximum obtained by this SDP as the \textit{value} of the round-$t$ relaxation, and say an infinite family of instances of MAX-$k$-XOR is hard for (or cannot be refuted by) $t$ rounds of Sum of Squares if there exists a constant $\mu$ such that every instance is at most $(1-\mu)$-satisfiable, but the round-$t$ SDP relaxation has value $1$. In other words, $t$-rounds of the SoS hierarchy cannot distinguish between completely satisfiable and $(1-\mu)$-satisfiable instances---this is often said to induce an \textit{integrality gap} for the problem of size $\frac{1}{1-\mu}$.

Rather than working directly with the Sum-of-Squares SDP relaxations, we prove our hardness results through a fruitful connection with refutation complexity due to Schoenebeck \cite{schoenebeck2008linear} and Tulsiani \cite{tulsiani2009csp}. More formally, following \cite{dinur2020explicit} we will use a proof system called $\oplus$-resolution where, given a system of linear equations $\Lambda$ over $\mathbb{F}_2$, we may derive new equations by mod 2 summation: 
\[
\{\ell_1 = b_1\}, \{\ell_2 =b_2\} \implies \ell_1 \oplus \ell_2 = b_1 \oplus b_2.
\]
A \textit{refutation} in this system is a derivation that $0=1$, and in our setting corresponds to a proof that the XOR instance given by $\Lambda$ is unsatisfiable. Schoenebeck \cite{schoenebeck2008linear} and Tulsiani \cite{tulsiani2009csp} showed that any system without a short refutation has a matching SoS lower bound. 
\begin{theorem}[{\cite[Lemma 13]{schoenebeck2008linear}} (as stated in \cite{dinur2020explicit})]\label{thm:resolution}
Let $\Lambda$ be a system of linear equations in $n$ variables over $\mathbb{F}_2$. If all refutations of $\Lambda$ have an equation using at least $2t$ variables, then the round-$t$ SoS Relaxation of $\Lambda$ has value $1$.
\end{theorem}

\subsection{Chain Complexes}\label{sec:chain}
While previous works constructing hard instances of CSPs rely on structure coming from graphs (e.g.\ \cite{grigoriev2001linear,schoenebeck2008linear}) or hypergraphs \cite{dinur2020explicit}, we take inspiration from recent work on c3-LTCs \cite{dinur2021,lin2022c} and qLDPC codes \cite{panteleev2021asymptotically,leverrier2022quantum} and instead study a more general set of objects called \textit{chain complexes}.
\begin{definition}[Chain Complex]
Let $X(0)$, $X(1)$, and $X(2)$ be sets, and $\partial_2: \mathbb{F}_2^{X(2)} \to \mathbb{F}_2^{X(1)}$, $\partial_1: \mathbb{F}_2^{X(1)} \to \mathbb{F}_2^{X(0)}$ linear maps. The sequence 
\[
X: \mathbb{F}_2^{X(0)} \overset{\partial_1}{\leftarrow} \mathbb{F}_2^{X(1)} \overset{\partial_2}{\leftarrow} \mathbb{F}_2^{X(2)}
\]
is called a (3-term) chain complex if $\partial_1\partial_2=0$.
\end{definition}
For the sake of intuition, let's take a moment to see why chain complexes are indeed a generalization of hypergraphs. Given an $r$-uniform hypergraph $H \subseteq {[n] \choose r}$, let $X(i) \subset {[n] \choose i}$ denote any $i$-set contained in some $r$-set in $H$. $H$ then induces an $(r+1)$-term chain complex:\footnote{Note $X(0)$ is defined to be the empty set, and that our indexing is off by $1$ from the usual notation in topology.}
\[
X: \mathbb{F}_2^{X(0)} \overset{\partial_1}{\leftarrow} \mathbb{F}_2^{X(1)} \overset{\partial_2}{\leftarrow} \ldots \overset{\partial_r}{\leftarrow} \mathbb{F}_2^{X(r)},
\]
where $\partial_if(x)$ is given by summing $f$ (mod 2) over $x$'s `boundary:'
\begin{equation}\label{eq:boundary-simp}
\forall f \in \mathbb{F}_2^{X(i)}: \partial_if(x) = \sum\limits_{y \in X(i): y \supset x}f(y).
\end{equation}
For instance, when $x$ is a vertex, $\partial_2f(x)$ averages over all edges containing $x$. As such, $\partial$ is usually called the \textit{boundary operator}, and it can be checked without too much difficulty that $\partial_{i-1}\partial_{i}=0$ (e.g.\ for $r=3$, this follows by noting a vertex is incident to either $0$ or $2$ edges of any given triangle).

In fact, the boundary operators can actually always be seen to have a similar form to \Cref{eq:boundary-simp}, even on a generic chain complex. This follows from passing to the matrix representation as discussed in \Cref{sec:intro}. Namely, we may view our 3-term chain complex as a pair of bipartite graphs $B_0=(X(0),X(1),E_1)$ and $B_1=(X(1),X(2),E_2)$, whose bipartite adjacency matrices are given by the matrix representations of $\partial_1$ and $\partial_2$ respectively (in the standard basis). In this setting, it is easy to see that $\partial_1$ and $\partial_2$ are also given by mod $2$ summation over neighbors on these underlying bipartite graphs:
\begin{align*}
\forall f \in \mathbb{F}_2^{X(1)}: \partial_1f(x) &= \sum\limits_{y \in X(1): (x,y) \in E_1} f(y) \quad (\text{mod} ~ 2)\\
\forall f \in \mathbb{F}_2^{X(2)}: \partial_2f(y) &= \sum\limits_{z \in X(2): (y,z) \in E_2} f(z) \quad (\text{mod} ~ 2),
\end{align*}
where we have assumed for simplicity that $\partial_1$ and $\partial_2$ are \textit{non-degenerate} in the sense that every row and column have at least one $1$.\footnote{In a graph, for instance, non-degeneracy corresponds to have no free-floating (degree $0$) vertices.} All complexes we study are non-degenerate, so we make this assumption throughout.

In matrix form, it is also easy to see that the transpose operators of $\partial$, called the \textit{co-boundary operators} and denoted $\delta_0 \coloneqq \partial_1^T$ and $\delta_1 \coloneqq \partial_2^T$, also form a chain complex in the opposite direction. As a result, we will usually write our chain complexes in the following form:
\[
X: \mathbb{F}_2^{X(0)} \overset{\delta_0}{\underset{\partial_1}{\rightleftarrows}} \mathbb{F}_2^{X(1)} \overset{\delta_1}{\underset{\partial_2}{\rightleftarrows}} \mathbb{F}_2^{X(2)}.
\]
We call elements of $\mathbb{F}_2^{X(i)}$ \textit{$i$-chains}, and note $\mathbb{F}_2^{X(i)}$ is often written as ``$C_i$'' in the literature. We avoid this notation since it conflicts with classical notation for codes used later in the paper.

Finally, before moving on to expansion on chain complexes, we cover two further concepts that will control important parameters of our corresponding XOR instaces: \textit{maximum degree} and \textit{explicitness}.
\begin{definition}[Maximum Degree]
The maximum degree of a chain complex $X: \mathbb{F}_2^{X(0)} \overset{\delta_0}{\underset{\partial_1}{\rightleftarrows}} \mathbb{F}_2^{X(1)} \overset{\delta_1}{\underset{\partial_2}{\rightleftarrows}} \mathbb{F}_2^{X(2)}$ is the maximum Hamming weight\footnote{The Hamming weight of binary vector $v$, denoted $\hnorm{v}$, counts the number of entries with a $1$.} across rows and columns of $\partial_1$ and $\partial_2$.
\end{definition}
In the bipartite graph view, this is simply the maximum vertex degree across both graphs. We call an infinite family of chain complexes $\textit{bounded degree}$ if there exists some constant $d \in \mathbb{N}$ such that all complexes in the family have maximum degree at most $d$.

Finally, in this work we will be interested in infinite families of chain complexes (and their associated XOR instances), so we need to define a notion of computational complexity over these objects. We will follow the standard notions used for expander families, and call a family of complexes \textit{explicit} if its elements can be constructed in deterministic polynomial time (this is often called \textit{mildly explicit}, but the difference is not particularly important in our setting).
\begin{definition}[Explicit Chain Complexes]
We call an infinite family of chain complexes $\{X_i\}$ explicit if there exists a determinstic algorithm computing each $X_i$ in time polynomial in $|X_i(0) \cup X_i(1) \cup X_i(2)|$.  
\end{definition}
All complexes studied in this work will be bounded-degree, in which case this notion may equivalently be defined looking only at the size of $X_i(0)$. This corresponds correctly to the standard notion of complexity for the associated $k$-CSP family where $|X_i(0)|$ gives the number of variables.
\subsection{Homology and High Dimensional Expansion}\label{sec:hom}
High dimensional expansion is a generalization of expansion in graphs originally introduced by Linial and Meshulam \cite{linial2006homological} (and later independently by Gromov \cite{gromov2010singularities}) to study the vanishing of homology in simplicial complexes. In this section we cover the basics of homology and introduce Linial and Meshulam's original notion of \textit{(co)-boundary expansion}. These notions (or modifications thereof) will play an important role in our CSP construction.

Following standard notation, we call functions in the kernel of $\partial_i$ \textit{cycles}, and functions in the kernel of $\delta_i$ \textit{co-cycles}, denoted:
\[
Z_i = \ker(\partial_i), \quad Z^i = \ker(\delta_i).
\]
Since $\delta^2=\partial^2=0$, notice that $\Ima(\partial_{i+1})$ are always cycles, and $\Ima(\delta_{i-1})$ are always co-cycles. We call functions in these classes \textit{boundaries} and \textit{co-boundaries} respectively, denoted:
\[
B_i = \Ima(\partial_{i+1}), \quad B^i = \Ima(\delta_{i-1}).
\]
The \textit{homology} and \textit{co-homology} of the chain complex correspond to (co)-cycles mod (co)-boundary:
\[
H_i = Z_i/B_i, \ \ H^i = Z^i/B^i,
\]
where $G/H$ denotes the quotient group. The notions of cycles and boundaries can be used to define a natural generalization of expander graphs to chain complexes called \textit{(Co)-boundary expansion}.
\begin{definition}[(Co)-Boundary Expansion]\label{def:boundary}
We call $X: \mathbb{F}_2^{X(0)} \overset{\delta_0}{\underset{\partial_1}{\rightleftarrows}} \mathbb{F}_2^{X(1)} \overset{\delta_1}{\underset{\partial_2}{\rightleftarrows}} \mathbb{F}_2^{X(2)}$ a $\rho$-boundary expander if the weight of any element in $\mathbb{F}_2^{X(1)} \setminus B_1$ is proportional to its distance from the boundary:
\[
\forall f \in \mathbb{F}_2^{X(1)} \setminus B_1: \frac{\hnorm{\partial_1 f}}{d(f,B_1)} \geq \rho,
\]
where $d(f,B_1)=\min_{b \in B_1}\hnorm{f+b}$. Similarly, $X$ is an $\rho$-co-boundary expander if:
\[
\forall f \in \mathbb{F}_2^{X(1)} \setminus B^1: \frac{\hnorm{\delta_1 f}}{d(f,B^1)} \geq \rho.
\]
\end{definition}
Since this definition may seem un-motivated at first glance, let's again take a look at the case of a graph $G=(V,E)$ which induces the (3-term) chain complex:
\[
X: \mathbb{F}_2^{\emptyset} \overset{\delta_0}{\underset{\partial_1}{\rightleftarrows}} \mathbb{F}_2^{V} \overset{\delta_1}{\underset{\partial_2}{\rightleftarrows}} \mathbb{F}_2^{E}.
\]
It is not hard to see that the co-boundary expansion of this chain is exactly Cheeger's constant:
\[
h(G) \coloneqq \min_{S \neq V,\emptyset}\left\{\frac{E(S,V\setminus S)}{\min\{|S|,|V \setminus S|\}}\right\},
\]
where $E(S,V\setminus S)$ is the standard notation for the size of the edge boundary between $S$ and the rest of the graph. This connection follows from noting that the only co-boundaries on this chain are $V$ and $\emptyset$, and that $\hnorm{\delta_1 1_S}$ exactly counts the edge-boundary of $S$, so in particular we have:
\[
\frac{\hnorm{\delta_1 1_S}}{d(1_S,B^1)} = \frac{E(S,V \setminus S)}{\min\{|S|,|V \setminus S|\}}.
\]
\section{Small Set Boundary Expansion}\label{sec:SS-boundary}

(Co)-boundary expansion is a very strong property, and unconditional construction of bounded degree (co)-boundary expanders is still a major open question in topological high dimensional expansion. Furthermore, (co)-boundary expansion actually implies the vanishing of (co)-homology. This is an issue in and of itself in our setting, since as discussed in \Cref{sec:pf-overview}, our CSP construction rests crucially on the associated chain complex having non-trivial co-homology. With this in mind, we introduce a new notion of high dimensional expansion which requires boundary expansion to hold \textit{only over small sets}.
\begin{definition}[Small-Set (Co)-Boundary Expansion]
We call $X: \mathbb{F}_2^{X(0)} \overset{\delta_0}{\underset{\partial_1}{\rightleftarrows}} \mathbb{F}_2^{X(1)} \overset{\delta_1}{\underset{\partial_2}{\rightleftarrows}} \mathbb{F}_2^{X(2)}$ a $(\rho_1,\rho_2)$-small-set boundary expander if the weight of small chains in $\mathbb{F}_2^{X(1)}$ is proportional to their distance from the boundary:
\[
\forall f \in \mathbb{F}_2^{X(1)} \setminus B_1, \hnorm{f} \leq \rho_1|X(i)|: \frac{\hnorm{\partial_1 f}}{d(f,B_1)} \geq \rho_2.
\]
Similarly, $X$ is a $(\rho_1,\rho_2)$-small-set co-boundary expander if:
\[
\forall f \in \mathbb{F}_2^{X(1)} \setminus B^1, \hnorm{f} \leq \rho_1|X(i)|: \frac{\hnorm{\delta_1 f}}{d(f,B^1)} \geq \rho_2
\]
We call $X$ a $(\rho_1,\rho_2)$-small-set HDX if it is both a $(\rho_1,\rho_2)$-small-set boundary and $(\rho_1,\rho_2)$-small-set co-boundary expander.
\end{definition}
Just like standard co-boundary expansion is a higher-order analog of Cheeger's constant (edge-expansion) in graphs, small-set co-boundary expansion is the natural analog of small-set expansion on graphs. Surprisingly, despite the recent prominence of small-set expansion in areas such as hardness of approximation (see e.g.\ \cite{raghavendra2010graph, subhash2018pseudorandom}), this simple generalization to higher dimensions seems to be missing from the literature even for the more standard notion of simplicial complexes (though as discussed in \Cref{sec:related-work} some similar notions have been studied towards building good co-systolic expanders \cite{kaufman2014ramanujan,evra2016bounded,kaufman2018cosystolic,lin2022c,lin2022good}). In this work we show how small-set (co)-boundary expanders can be transformed into explicit hard CSP instances for linear levels of Sum-of-Squares. Given the prominence of small-set expansion throughout hardness of approximation, we expect SS-HDX may have many further applications.

Before moving on, it will be useful to observe two important implications of a complex satisfying small-set (co)-boundary expansion. First, while the notion does not require the vanishing of (co)-homology like standard boundary expansion, it does still imply a strong restriction on the structure of elements in $Z_1 \setminus B_1$: they must be large.
\begin{lemma}[Small-Set (Co)-Boundary Expansion $\to$ (Co)-Systolic Distance]\label{lem:ss-boundary-to-systolic-distance}
If $X$ is a $(\rho_1,\rho_2)$-small-set boundary expander, then all chains $f \in Z_1 \setminus B_1$ are large:
\begin{equation}\label{eq:sys-distance}
\min_{f \in Z_1 \setminus B_1 }\left\{\hnorm{f} \right\} > \rho_1|X(1)|.
\end{equation}
Similarly, if $X$ is a $(\rho_1,\rho_2)$-small-set co-boundary expander, then all chains $f \in Z^1 \setminus B^1$ are large:
\begin{equation}\label{eq:cosys-distance}
\min_{f \in Z^1 \setminus B^1 }\left\{\hnorm{f} \right\} > \rho_1|X(1)|.
\end{equation}
\end{lemma}
\begin{proof}
We prove the first statement only, the second follows similarly. Assume $h_1 \in Z_1 \setminus B_1$ satisfies $\hnorm{h_1} \leq \rho_1|X(1)|$. Since $h_1$ is a cycle, we have $\partial_1 h_1 = 0,$ but then by small-set boundary expansion we have $d(h_1,B_1) = 0$, so $h_1 \in B_1$ giving the desired contradiction.
\end{proof}
We say complexes satisfying \Cref{eq:sys-distance} have \textit{systolic distance} $\rho_1$, and complexes satisfying \Cref{eq:cosys-distance} have \textit{co-systolic distance} $\rho_1$. As discussed in \Cref{sec:pf-overview}, these properties were recently crucial to the construction of good qLDPC codes \cite{panteleev2021asymptotically}, and were also used by \cite{dinur2020explicit} to prove the soundness of their 3-XOR construction. Indeed it is worth noting that bounded co-systolic distance is actually enough for soundness in our construction as well, we only truly need the full power of small-set boundary expansion in one direction.

Second, we will crucially rely on a standard connection between boundary expansion and a concept known as an \textit{isoperimetric inequality}, which relates the size of an object to the size of its boundary.\footnote{For example the isoperimetric inequality on $\mathbb{R}^2$ says the length (boundary) of any closed curve is at least $2\sqrt{\pi}$ times the square root of its area.} In particular, it is well known that boundary expansion is actually equivalent to an isoperimetric inequality for \textit{minimal} chains (see e.g.\ \cite{kaufman2016isoperimetric}).
\begin{definition}[Minimal Chains]\label{def:minimal}
A function $h \in \mathbb{F}_2^{X(1)}$ is called minimal if $\forall b \in B_1$, $\hnorm{h+b} \leq \hnorm{h}$.
\end{definition}
A similar equivalence holds for small-set boundary expansion as well, and will be crucial for the completeness of our CSP instances: $X$ is a small-set boundary expander if and only if small, minimal chains in $X$ satisfy an isoperimetric inequality.
\begin{lemma}[Small-Set (Co)-Boundary $\leftrightarrow$ (Co)-Isoperimetric Inequality]\label{lemma:SS-LTC}
Let $X$ be a $(\rho_1,\rho_2)$-small-set boundary expander. Then for any $h \in \mathbb{F}_2^{X(1)}$ satisfying:
\begin{enumerate}
    \item $h$ is small: $\hnorm{h} \leq \rho_1|X(1)|$
    \item $h$ is minimal: $\forall b \in B_1: \hnorm{h+b} \geq \hnorm{h}$
\end{enumerate}
the boundary $\partial_1 h$ must be large relative to $h$:
\begin{equation}\label{eq:iso}
\hnorm{\partial_1 h} \geq \rho_2 \hnorm{h}.
\end{equation}
Conversely if \Cref{eq:iso} holds for any small minimal chain, then $X$ is a $(\rho_1,\rho_2)$-small-set boundary expander.
\end{lemma}
\begin{proof}
We start with the forward direction. Since $\hnorm{h} \leq \rho_1|X(i)|$ and $h$ is minimal, by small-set boundary expansion we have that:
\[
\hnorm{\partial_1 h} \geq \rho_2 d(h,B) = \rho_2\min_{b \in B}\{\hnorm{h+b}\} = \rho_2\hnorm{h}.
\]
The converse implication is similar. Let $h \in \mathbb{F}_2^{X(1)}$ be a small chain satisfying $|h| \leq \rho_1|X(1)|$, and let $b \in B_1$ be the boundary minimizing $\hnorm{h+b}$. Then by isoperimetry of $h+b$, we have:
\[
\hnorm{\partial_1 h} = \hnorm{\partial_1(h+b)} \geq \rho_2\hnorm{h+b} = \rho_2d(h,B_1)
\]
as desired.
\end{proof}
We note the same result holds for co-boundary expansion by the same proof. Isoperimetry (combined with good systolic distance) will be crucial for showing completeness of our XOR instances, replacing the use of Gromov's filling inequality in \cite{dinur2020explicit}. 

\section{From Expansion to Hardness}\label{sec:HDX-to-XOR}
We now show how to translate any family of expanding, bounded-degree 3-term chain complexes with non-trivial cohomology into hard instances of 3-XOR for $\Omega(n)$-levels of Sum-of-Squares. 
\begin{theorem}\label{thm:HDX-to-XOR}
Let $\{X_i\}$ be an explicit family of chain complexes of maximum degree $k\in \mathbb{N}$ and $\mu,\rho_1,\rho_2 \in (0,1)$ constants such that:
\begin{enumerate}
    \item $H^1$ is non-trivial,
    \item $X$ has $\mu$-co-systolic distance,
    \item $X$ is a $(\rho_1,\rho_2)$-small-set boundary expander.
\end{enumerate}
Then there exist constants $\mu_1,\mu_2 \in (0,1)$ depending only on $k$, $\mu$, $\rho_1$, and $\rho_2$ and an explicit family of MAX-$3$-XOR instances $\{\mathcal{I}_i\}$ on $n_i$ variables such that:
\begin{enumerate}
    \item Every instance is at most $(1-\mu_1)$-satisfiable,
    \item No instance can be refuted by $\mu_2n_i$ levels of the SoS hierarchy.
\end{enumerate}
Moreover if the complex has degree lower bounded by $3$, $\{\mathcal{I}_i\}$ are instances of $3$-XOR.
\end{theorem}

\Cref{thm:HDX-to-XOR} is actually proved mainly by associating an instance of MAX-$k$-XOR to every complex $X_i$ in the family. Moving to $3$-XOR can then be done through standard NP-reduction arguments within the SoS hierarchy.\footnote{Though one must be careful that the number of variables does not blow up in the reduction, as we discuss later in the section.} Thus the main challenge is to build hard instances of MAX-$k$-XOR from our complexes. We'll start by overviewing our construction, which is a generalization of \cite{dinur2020explicit}'s $3$-XOR construction from simplicial complexes to generic chain complexes.

\paragraph{Construction:} It will be convenient to phrase our construction in the bipartite graph formulation discussed in \Cref{sec:prelims}. Recall that any chain complex $X: \mathbb{F}_2^{X(0)} \overset{\delta_0}{\underset{\partial_1}{\rightleftarrows}} \mathbb{F}_2^{X(1)} \overset{\delta_1}{\underset{\partial_2}{\rightleftarrows}} \mathbb{F}_2^{X(2)}$ may be written as a pair of bipartite graphs $B_1=(X(0),X(1),E_1)$ and $B_2=(X(1),X(2),E_2)$ where $E_1$ and $E_2$ are uniquely determined by the matrix representations of the boundary operators. Assuming our complex has non-trivial co-homology, let $\beta \in Z^1 \setminus B^1$.\footnote{Note that $\beta$ can be found in polynomial time by standard linear algebraic techniques.} Our associated CSP $\mathcal{I}_{X,\beta}$ is given by adding for every $y \in X(1)$ the constraint:
\[
C_y \coloneqq \left \{\sum\limits_{x \in X(0): (x,y) \in E_1} x \ \ (\text{mod} ~ 2) = \beta(x)\right\}.
\]
Since the choice of $\beta \in Z^1 \setminus B^1$ will not matter, in what follows we will drop it from the notation and just write $\mathcal{I}_X$. We make two observations about $\mathcal{I}_X$ before moving on. First, let's confirm $\mathcal{I}_X$ is indeed an instance of MAX-$k$-XOR.
\begin{observation}
If $X$ has maximum degree $k$, then $\mathcal{I}_X$ is an instance of MAX-$k$-XOR.
\end{observation}
\begin{proof}
This follows immediately from the chain complex having maximum degree $k$, as every $y \in X(1)$ then has at most $k$ neighbors in $X(0)$ (i.e.\ that there are at most $k$ elements $x$ such that $(x,y) \in E_1$).
\end{proof}
Second, we observe that our instances have at most a linear number of constraints.
\begin{observation}
If $X$ has maximum degree $k$, then $\mathcal{I}_X$ has at most $k|X(0)|$ constraints.
\end{observation}
\begin{proof}
Since our complex is non-degenerate and degree at most $k$, we have that $|X(1)| \leq k|X(0)|$. $\mathcal{I}_X$ has $|X(1)|$ constraints by construction.
\end{proof}
As a result, any explicit infinite family of bounded degree chain complexes with non-trivial cohomology induces an explicit infinite family of MAX-$k$-XOR instances with linearly many constraints for some constant $k \in \mathbb{N}$. The main work in proving \Cref{thm:HDX-to-XOR} therefore boils down to proving that the instances $\mathcal{I}_X$ are \textit{sound} (at most $(1-\mu)$-satisfiable), and \textit{complete} (look satisfiable to SoS). 
\begin{theorem}\label{thm:HDX-to-kXOR}
Let $X$ be a chain complex of maximum degree $k$ and $\mu,\rho_1,\rho_2 \in (0,1)$ constants such that:
\begin{enumerate}
    \item $H^1$ is non-trivial,
    \item $X$ has $\mu$-co-systolic distance,
    \item $X$ is a $(\rho_1,\rho_2)$-small-set boundary expander.
\end{enumerate}
Then $\mathcal{I}_X$ is an instance of MAX-$k$-CSP on $|X(0)|$ variables satisfying:
\begin{enumerate}
    \item Soundness: $\mathcal{I}_X$ is at most $(1-\mu)$-satisfiable,
    \item Completeness: $\mathcal{I}_X$ cannot be refuted by $\left(\frac{\rho_1\rho_2}{4k}|X(0)|\right)$-levels of the SoS hierarchy.
\end{enumerate}
\end{theorem}
\noindent We'll break the proof of \Cref{thm:HDX-to-kXOR} into two parts, corresponding to soundness and completeness. 

\paragraph{Soundness:} The soundness of our construction can be proved with no further background, and is a direct generalization of arguments in \cite{dinur2020explicit} from simplicial complexes to general chain complexes.
\begin{proof}[Proof of Soundess (\Cref{thm:HDX-to-kXOR})]
Recall that our constraints are defined by some function $\beta \in Z^1 \setminus B^1$. Let $f \in \mathbb{F}_2^{X(0)}$ be a potential assignment to variables in our instance. For any constraint $y \in X(1)$, we can check if $f$ satisfies $y$ by evaluating $(\beta + \delta_0 f)(y)$:
\[
(\beta + \delta_0 f)(y) = \beta(y) + \sum\limits_{(x,y) \in E_1} f(x).
\]
In other words, the Hamming weight $\hnorm{\beta+\delta_0 f}$ exactly corresponds to the number of violated constraints in our instance. The key is now to observe that since $\beta \in Z^1 \setminus B^1$, $\beta + \delta_0 f$ also lies in $Z^1 \setminus B^1$. Since $X$ has $\mu$-co-systolic distance, we have $\hnorm{\beta + \delta_0f} \geq \mu|X(1)|$, so any assignment to variables must violate at least a $\mu$ fraction of constraints as desired.
\end{proof}

\paragraph{Completeness:} Proving the completeness of \Cref{thm:HDX-to-kXOR} requires a bit more setup. As discussed in \Cref{sec:prelims}, we appeal to the general paradigm of Grigoriev \cite{grigoriev2001linear}, Schoenebeck \cite{schoenebeck2008linear}, and Tulsiani \cite{tulsiani2009csp} relating refutation width with Sum-of-Squares completeness. Our lower bound on the refutation width of $\mathcal{I}_X$ can be viewed in some sense as a mix of the classical strategy of Ben-Sasson and Wigderson \cite{ben1999short} (who used traditional boundary expansion on graphs to show lower bounds against refuting Tseiten formulas) and the recent argument of \cite{dinur2020explicit} using Gromov's filling inequality on the Ramanujan complex. We mostly follow the exposition given in the latter.

We will consider refutations in the $\oplus$-resolution proof system, in which two linear equations $\ell_1=b_1$ and $\ell_2=b_2$ can be added to derive $\ell_1 \oplus \ell_2 = b_1 \oplus b_2$. By \Cref{thm:resolution}, it is enough to prove that any refutation of the linear equations corresponding to $\mathcal{I}_X$ has width at least $\frac{\rho_1\rho_2}{2k}|X(0)|$, where width measures the largest number of variables appearing in any equation in the refutation. A refutation in the $\oplus$-resolution proof system can be modeled as a DAG where leaves correspond to linear equations (our XOR constraints), internal nodes have two incoming edges and correspond to the XOR of their parents, and the root derives the contradiction $0=1$. 

To track the number of variables at each step, we follow the strategy of \cite{dinur2020explicit} and associate to each node $v$ of the DAG a function $h_v \in \mathbb{F}_2^{X(1)}$ and value $b_v \in \{0,1\}$ as follows. Since each leaf in the refutation corresponds to one of our XOR constraints, assign the leaf corresponding to $s \in X(1)$ the indicator $1_s \in \mathbb{F}_2^{X(1)}$ and value $\beta(s) \in \mathbb{F}_2$ (where we recall $\beta \in Z^1\setminus B^1$ was the chain used to define our constraint values). The function and value assigned to each internal node $v$ with parents $v_1,v_2$ is then defined recursively to be the (mod 2) sum of its parents:
\[
h_v = h_{v_1} \oplus h_{v_2},~\text{and }~\beta_v = \beta_{v_1} \oplus \beta_{v_2}.
\]
Notice that by construction, $\partial_1 h_v$ exactly corresponds to the variables appearing in the linear equation at node $v$. This means we can bound the width of the refutation by identifying some node $v$ in the refutation whose associated function $h_v$ has large boundary.

To this end, following \cite{dinur2020explicit}'s high dimensional variant of \cite{ben1999short}'s original technique we define the following potential function across nodes in our refutation:
\[
\kappa(v) \coloneqq \min_{b \in B_1}\hnorm{h_v + b}.
\]
Our goal will be to find a node in the refutation whose potential is large, but still small enough that we can apply small-set boundary expansion. Namely, if we can find $v$ such that $\frac{\rho_1}{2}|X(1)| \leq \kappa(v) \leq \rho_1|X(1)|$, then by our isoperimteric inequality for small sets (\Cref{lemma:SS-LTC}) we have:
\[
\hnorm{\partial h_v} = \hnorm{\partial(h_v + b)} \geq \rho_2\hnorm{h_v + b} \geq \frac{\rho_1\rho_2}{2}|X(1)| \geq \frac{\rho_1\rho_2}{2k}|X(0)|
\]
which would give the desired bound on refutation width. With this in mind, we can finally prove completeness.
\begin{proof}[Proof of completeness (\Cref{thm:HDX-to-kXOR})]
As discussed above, it is sufficient to prove that any refutation has width at least $\frac{\rho_1\rho_2}{2k}|X(0)|$, and that this can be done by finding a node $v$ with potential $\frac{\rho_1}{2}|X(1)| \leq \kappa(v) \leq \rho_1|X(1)|$. The proof follows the classical strategy of \cite{ben1999short}. Namely it is enough to show the following three properties:
\begin{enumerate}
    \item The root node has large potential: $\kappa(r) > \rho_1|X(1)|$
    \item The leaves have small potential: $\kappa(s) \leq 1$
    \item The potential function is sub-additive: $\kappa(v) \leq \kappa(v_1)+\kappa(v_2)$.
\end{enumerate}
As long as these hold, getting from the leaf potential of (at most) $1$ to the root potential of $\kappa(r)  >\rho_1|X(1)|$ requires passing through some internal node $v$ with $\frac{\rho_1}{2}|X(1)| \leq \kappa(v) \leq \rho_1|X(1)|$ as desired.

It is left to prove the three properties, which follow from similar analysis as in \cite{dinur2020explicit} for the Ramanujan complex. The second and third properties are essentially immediate. Leaves are given by the indicator function of elements $s \in X(1)$, which are at most distance one from $\vec{0} \in B_1$ (the all $0$s function). Sub-additivity follows from the triangle inequality. For a node $v$ with parents $v_1$ and $v_2$, let $b_1$ and $b_2$ be boundaries minimizing $d(h_{v_1},B_1)$ and $d(h_{v_2},B_1)$, then we have:
\[
\kappa(v_1)+\kappa(v_2) = \hnorm{h_{v_1} + b_1} + \hnorm{h_{v_2}+b_2} \geq \hnorm{h_{v_1} + b_1+ h_{v_2}+b_2} = \hnorm{h_{v} + b_1+b_2} \geq \kappa(v).
\]
For the first property, we argue the root node $r$ must satisfy $h_r \in Z_1 \setminus B_1$. If this is the case we are done by the fact that our complex has good co-systolic distance by \Cref{lem:ss-boundary-to-systolic-distance}:
\[
\kappa(h_r) = \min_{b \in B_1}\hnorm{h_r + b} > \rho_1|X(1)|,
\]
since any $h_r+b \in Z_1 \setminus B_1$ as well. To see that $h_r \in Z_1 \setminus B_1$, first note that since the root node in our refutation corresponds to the equation $0=1$, we must have $\partial_1 h_r = 0$ and therefore $h_r \in Z_1$. To complete the proof we therefore only need to show $h_r \notin B_1$, which follows from the fact that $b_r=1$ for the root node. Namely, notice that for any node $v$ we have
$b_v = \langle \beta, h_v \rangle$
by construction (since we are just summing mod $2$ over the constraints), and in particular that $\langle \beta, h_r \rangle = 1$. On the other hand, if $h_r \in B_1$, then by definition there exists $f \in \mathbb{F}_2^{X(2)}$ such that $h_r=\partial_2 f$ and since $\delta_1 = \partial_2^T$ we have
\[
\langle \beta, h_r \rangle = \langle \beta, \partial_2 f \rangle = \langle \delta_1 \beta, f \rangle = 0
\]
since $\beta \in Z^1$. Thus $h_r$ is in $Z_1$ but not $B_1$, which completes the proof.
\end{proof}

We are now one step away from proving \Cref{thm:HDX-to-XOR}; we just need to show how to move from a hard instance of MAX-$k$-XOR to a hard instance of $3$-XOR. Such a reduction is fairly standard within the SoS literature, but we'll include the proof for completeness. To do so, we'll need to introduce a second way to characterize completeness of an instance for $t$ rounds of SoS through an object called a \textit{pseudo-expectation}. Given a set of variables $\{x_i\}_{i \in [n]}$ and $d \in \mathbb{N}$, let  $\text{poly}_{\mathbb{R}}(\{x_i\},d)$ denote the set of degree at most $d$ polynomials in $\R[x_1,\ldots,x_n]$. For our purposes, it is enough to think of a degree $2t$ pseudo-expectation as an operator $\tilde{\mathbb{E}}: \text{poly}_{\mathbb{R}}(\{x_i\},2t) \to \mathbb{R}$ that `pretends' to be an expectation in the following four ways:
\begin{enumerate}
    \item Scaling: $\tilde{\mathbb{E}}[1]=1$
    \item Linearity: 
    \[
    \forall a,b \in \mathbb{R}, p(x),q(x) \in \text{poly}_{\mathbb{R}}(\{x_i\},2t): \ \ \tilde{\mathbb{E}}[ap(x)+bq(x)]=a\tilde{\mathbb{E}}[p(x)]+b\tilde{\mathbb{E}}[q(x)]
    \]
    \item Positivity of Squares: 
    \[
    \forall q(x) \in \text{poly}_{\mathbb{R}}(\{x_i\},t): \ \  \tilde{\mathbb{E}}[q(x)^2] \geq 0
    \]
    \item Booleanity: 
    \[\forall j \in [n], p(x) \in \text{poly}_{\mathbb{R}}(\{x_i\},2t-2): \ \  \tilde{\mathbb{E}}[x_j^2p(x)] =\tilde{\mathbb{E}}[p(x)].
    \]
\end{enumerate}
With this in mind, let $\mathcal{I}$ be an instance of XOR on $n$ variables $\{x_1,\ldots,x_n\}$. It will be convenient to express constraints in $C_i \in \mathcal{I}$ multiplicatively as:
\[
C_i \coloneqq \left\{x_{i_1}\ldots x_{i_j}=b_i \right\}
\]
where $b_i \in \{-1,1\}$ and assignments now range over $\{-1,1\}^n$. Let $C_i(x)$ be shorthand for the lefthand product of variables in the constraint, and $|C_i|$ denote the degree of $C_i(x)$. It turns out (see e.g.\ \cite{fleming2019semialgebraic}) that completeness of $\mathcal{I}$ against $t$ levels of Sum-of-Squares is equivalent to the existence of a degree $2t$ pseudo-expectation which respects every constraint $C_i \in \mathcal{I}$ in the following strong sense:
\begin{equation}\label{eq:SoS-constraints}
\forall p(x) \in \text{poly}_{\mathbb{R}}(\{x\},2t-|C_i|): \ \ 
\tilde{\mathbb{E}}[C_i(x)p(x)] = b_i\tilde{\mathbb{E}}[p(x)].
\end{equation}
With this in mind, we can finally put everything together and prove \Cref{thm:HDX-to-XOR}.
\begin{proof}[Proof of \Cref{thm:HDX-to-XOR}]
We'll start by constructing an explicit family of hard instances of MAX-$k$-XOR, then reduce to $3$-XOR through the above machinery. By \Cref{thm:HDX-to-kXOR}, every complex $X_i$ in our family corresponds to an instance $\mathcal{I}_{X_i}$ of MAX-$k$-XOR  on $n_i=|X_i(0)|$ vertices and $m_i \leq k|X(0)|$ constraints that is at most $(1-\mu)$-satisfiable but cannot be refuted by the $\frac{\rho_1\rho_2}{4k}n_i$-level SoS relaxation. Furthermore each instance $\mathcal{I}_{X_i}$ can be constructed in $\text{poly}(n_i)$ time. This follows immediately from the fact that $\{X_i\}$ itself is explicit (and bounded degree), and that finding some $\beta \in Z_1 \setminus B_1$ can be done in polynomial time by basic linear algebra over dimension $O(n_i)$ vector spaces. 

It is left to argue that we can use $\mathcal{I}_{X_i}$ to construct a corresponding instance of $3$-XOR that remains hard for Sum-of-Squares. We will use the following simple approach: given a clause with more than $3$ variables, split it into two clauses of about half the size whose product is the original clause. More formally, given a constraint $C_i \coloneqq \left\{x_{i_1}\ldots x_{i_j}=b_i\right\}$, we apply the transformation:
\begin{equation}\label{eq:reduction}
C_i \to \left\{C_i^{(0)} \coloneqq \{x_{i_1}\ldots x_{i_{\floor{j/2}}}y_i = b_i\}, \ \ C_i^{(1)} \coloneqq \{x_{i_{\floor{j/2}+1}}\ldots x_{i_{j}}y_i = 1\}\right\}
\end{equation}
where $y_i$ is a newly introduced `dummy' variable. Given a generic instance of MAX-$k$-XOR $\mathcal{I}_k$, let $\Phi(\mathcal{I}_k)$ denote the CSP resulting from applying the above transformation to every constraint with more than $3$ variables. We will argue that $\Phi(\mathcal{I}_k)$ has about half as many variables per clause as the original instance, but maintains soundness and completeness up to constant factors.
\begin{claim}\label{claim:reduction}
Let $\mathcal{I}_k$ be an instance of MAX-$k$-XOR for $k \geq 4$ on $n$ variables and $m$ constraints such that:
\begin{enumerate}
    \item $\mathcal{I}_k$ is at most $(1-\mu)$-satisfiable,
    \item $\mathcal{I}_k$ cannot be refuted by $t$ rounds of Sum-of-Squares.
\end{enumerate}
Then $\Phi( \mathcal{I}_k )$ is an instance of MAX-$j$-XOR for $j = \ceil{k/2}+1$ on at most $n+m$ variables and $2m$ constraints satisfying:
\begin{enumerate}
    \item $\mathcal{I}_k$ is at most $(1-\mu/2)$-satisfiable,
    \item $\mathcal{I}_k$ cannot be refuted by $\frac{2t}{k}$ rounds of Sum-of-Squares.
\end{enumerate}
\end{claim}
Let's first show \Cref{claim:reduction} completes the proof of our main theorem. Starting from our MAX-$k$-XOR instance $\mathcal{I}_{X_i}$, \Cref{claim:reduction} shows that $\Phi^{\ceil{\log(k)}}(\mathcal{I}_{X_i})$ is an instance of MAX-$3$-XOR on $O_k(n_i)$ variables that is at most $(1-\Omega_k(\mu))$-satisfiable but cannot be refuted by $\Omega_k(n_i)$ rounds of Sum-of-Squares. This follows from the fact that the original (and all transformed instances) have $m \leq O_k(n_i)$ constraints.\footnote{We note that it is possible to improve the dependence on $k$ by slightly more involved analysis, but since $k$ is just a constant we choose to work with iterated applications of the above for simplicity of exposition.} Finally, if the original instance had no constraints with fewer than $3$ variables (which occurs if the original complex has degree lower bounded by $3$), $\Phi^{\ceil{\log(k)}}(\mathcal{I}_{X_i})$ is an instance of $3$-XOR. With this in mind, it is left to prove the claim.

\begin{proof}[Proof of \Cref{claim:reduction}]
The fact that $\Phi(\mathcal{I}_k)$ is an instance of MAX-$j$-XOR for $j = \ceil{k/2}+1$ on at most $n+m$ variables and at most $2m$ constraints is immediate from construction. The main interest lies in proving soundness and completeness of the instance.
\\
\\
\noindent\textbf{Soundness:} Soundness of $\Phi(\mathcal{I}_k)$ follows from observing that since $y_i^2=1$, $C_i = C_i^{(0)}\cdot C_i^{(1)}$. Namely by the soundness of the original instance, any assignment of variables to $\Phi(\mathcal{I}_k)$ must fail at least a $\mu$ fraction of original constraints $C_i$ (since these have no dependence on the new dummy variables). If $C_i=C_i^{(0)}\cdot C_i^{(1)}$ is violated it must be the case that either $C_i^{(0)}$ or $C_i^{(1)}$ is violated, so any assignment of variables to our transformed CSP $\Phi(\mathcal{I}_k)$ must still violate at least a $\mu/2$ fraction of its constraints.
\\
\\
\noindent\textbf{Completeness:} Given a degree $2t$ pseudo-expectation $\tilde{\mathbb{E}}$ satisfying the constraints of $\mathcal{I}_k$ (in the sense of \Cref{eq:SoS-constraints}), we must construct a new pseudo-expectation $\tilde{\mathbb{E}}_{\Phi}$ on the variables of $\Phi(\mathcal{I}_k)$ satisfying the transformed constraints. Given a polynomial $p(x,y) \in \R[\{x_i\},\{y_j\}]$, let $p(x,1) \in \R[\{x_i\}]$ denote the result of setting each $y$ variable to $1$. The idea is to observe that each dummy variable $y_i$ in the new instance can really be thought of as a `stand-in' for the product $b_ix_{i_1}\ldots x_{i_{\floor{j/2}}}=b_iC_i^{(0)}(x,1)$ in the sense that replacing each $y_i$ with $b_iC_i^{(0)}(x,1)$ simply returns the original instance. This suggests a natural strategy for defining our new pseudo-expectation $\tilde{\mathbb{E}}_{\Phi}$: just replace $y_i$ with $b_iC_i^{(0)}(x,1)$.\footnote{We thank Sam Hopkins for suggesting this general approach.}

Formally, this takes a bit of work. Let $S \subseteq [m]$ denote the set of indices on which we transformed our original instance, $\{y_j\}_{j \in S}$ denote the newly introduced variables, and $T: \R[\{x_i\}_{i \in [n]}, \{y_j\}_{j \in S}] \to \R[x_1,\ldots,x_n]$ denote the map which independently replaces each occurrence of $y_j$ with $b_jC_j^{(0)}(x,1)$ (and leaves variables in $\{x_i\}$ unchanged). It is an elementary exercise to show that $T$ satisfies the following useful properties:
\begin{enumerate}
    \item $T$ is (additively) linear: \[T(az(x,y))=aT(z(x,y)) \ \text{ and } \ T(z_1(x,y) + z_2(x,y)) = T(z_1(x,y)) + T(z_2(x,y))\]
    \item $T$ is (multiplicatively) linear: \[T(z_1(x,y)z_2(x,y)) = T(z_1(x,y))T(z_2(x,y))\]
    \item $T$ does not substantially blow up degree: \[Deg(T(z(x,y))) \leq \floor{k/2} Deg(z(x,y)).\]
\end{enumerate}
With this in mind, define the value of our new pseudo-expectation on any degree at most $\frac{2t}{\floor{k/2}}$ polynomial $z(x,y) \in \R[\{x_i\}_{i \in [n]}, \{y_j\}_{j \in S}]$ as:
\[
\tilde{\mathbb{E}}_{\Phi}[z(x,y)] \coloneqq \tilde{\mathbb{E}}[T(z(x,y))]
\]
which is well-defined by the third property. It is an easy exercise to check that $\tilde{\mathbb{E}}_{\Phi}$ remains a pseudo-expectation, as the linearity of $T$ ensures scaling, linearity, positivity of squares, and booleanity are all inherited from $\tilde{\mathbb{E}}$. Thus it is left to check that $\tilde{\mathbb{E}}_{\Phi}$ satisfies every constraint $C_i^{(j)} \in \Phi(\mathcal{I}_k)$ in the sense of \Cref{eq:SoS-constraints}. To see this, first observe that
\begin{align*}
    \tilde{\mathbb{E}}_\Phi[C_i^{(j)}(x,y)z(x,y)] &= \tilde{\mathbb{E}}[T(C_i^{(j)}(x,y)z(x,y))]\\
    &= \tilde{\mathbb{E}}[T(C_i^{(j)}(x,y))T(z(x,y))].
\end{align*}
Taking a closer look at $T(C_i^{(j)}(x,y))$, we have by definition that:
\[
T(C_i^{(j)}(x,y)) = \begin{cases}
b_iC_i(x) & \text{if } j=1\\
b_i(C_i^{(0)}(x,1))^2 & \text{if } j=0.
\end{cases}
\]
Breaking into case analysis, we then have for $j=1$:
\begin{align*}
\tilde{\mathbb{E}}[T(C_i^{(1)}(x,y))T(z(x,y))] &= \tilde{\mathbb{E}}[b_iC_i(x)T(z(x,y))]\\
&=\tilde{\mathbb{E}}[T(z(x,y))]\\
&=\tilde{\mathbb{E}}_\Phi[z(x,y)]
\end{align*}
and for $j=0$ that:
\begin{align*}
\tilde{\mathbb{E}}[T(C_i^{(0)}(x,y))T(z(x,y))] &= \tilde{\mathbb{E}}[b_i(C^{(0)}_i(x,1))^2T(z(x,y))]\\
&=b_i\tilde{\mathbb{E}}[T(z(x,y))]\\
&=b_i\tilde{\mathbb{E}}_\Phi[z(x,y)]
\end{align*}
which match the form of the constraints given in \Cref{eq:reduction} as desired.
\end{proof}
\end{proof}

\section{Preliminaries II: Constructing SS-HDX} \label{sec:prelims2}

We now cover the tools necessary for constructing our small-set HDX, including background on basic expander graphs, left-right Cayley complexes, error correcting codes, Tanner codes, and tensor codes. We closely follow the discussion in \cite{leverrier2022quantum} who largely cover the same background material.

\subsection{Expander Graphs}
The main building block of Leverrier and Z\'emor's qLDPC codes are a ubiquitous class of graphs in computer science called \textit{spectral expanders}. Let $\cG = (V, E)$ be an undirected $\Delta$-regular (multi)-graph on $n$ vertices, and define $\lambda(\cG) := \max\{|\lambda_2|, |\lambda_n|\}$ where $\Delta=\lambda_1 \ge \lambda_2 \ge ... \ge \lambda_n$ are the eigenvalues of the adjacency matrix of $G$. We say $G$ is a $\lambda$-spectral expander if $\lambda(\cG) \leq \lambda$, and call it \textit{Ramanujan} if $\lambda(\cG) \le 2 \sqrt{\Delta-1}$, which is the optimal expansion for infinite families of fixed degree \cite{alon1986eigenvalues}.

We will rely on spectral expanders for two main reasons. First, as we will discuss in the following section, infinite families of these objects are well-known not only to exist, but to be explicitly constructable (see e.g.\ \cite{morgenstern1994existence}). Second, spectral expansion provides a useful proxy for edge-expansion in the sense that for any $S,T \subseteq V$, there cannot be too many edges passing between $S$ and $T$. This is classically known as the \textit{expander-mixing lemma}, and likely first appeared in \cite{alon1988explicit}:
\begin{lemma} [Expander mixing lemma]
    Let $\cG$ be a $\Delta$-regular graph. Then for any subset $S, T \subset V(G)$ we have
    \begin{equation*}
        |E(S, T)| \le \frac{\Delta}{|V|} |S||T| + \lambda(\cG) \sqrt{|S||T|}.
    \end{equation*}
\end{lemma}

When $|S|$ and $|T|$ are small compared with $|V|$, we will think of $\lambda(\cG) \sqrt{|S||T|}$ as the main term and $\frac{\Delta}{|V|} |S||T|$ as the error term (we note this is the opposite of how the lemma is often applied).

It will also be important for us that the expander mixing lemma holds for \textit{double covers} of a spectral expanders with a small modification. The double cover $\cG' = (V', E')$ of a graph $\cG=(V, E)$ has vertex set $V' = V_0 \cup V_1$, for $V_0 = V \times \{0\}$ and $V_1 = V \times \{1\}$, and edge set $E' = \{\{(v, 0), (w, 1)\}: v, w \in V, \{v, w\} \in E\}$. The expander mixing lemma applies for double covered graphs when $S \subset V_0, T \subset V_1$.
\begin{lemma} [Expander mixing lemma for double covered graph]
    Let $\cG$ be a $\Delta$-regular graph and $\cG'$ be its double cover. Then for any subset $S \subset V_0(\cG'), T \subset V_1(\cG')$ we have
    \begin{equation*}
        |E(S, T)| \le \frac{\Delta}{|V(\cG)|} |S||T| + \lambda(\cG) \sqrt{|S||T|}.
    \end{equation*}
\end{lemma}

This can be shown easily by projecting $S$ and $T$ back to the original graph.

\subsection{Left-Right Cayley Complexes}\label{sec:cayley}
While expansion is a useful property in its own right, our arguments require higher dimensional structure. The key lies in an object called the \textit{left-right Cayley complex} introduced in \cite{dinur2021} to build c3-LTCs. A left-right Cayley complex is determined by a group $G$ and two sets of generators $A = A^{-1}$ and $B = B^{-1}$. The complex consists of vertices, $A$-edges, $B$-edges, and squares as follows:
\begin{itemize}
    \item The vertices are $V^0 = G$.
    \item The $A$-edges are $E^0_A$ and the $B$-edges are $E^0_B$ where
        \begin{equation*}
            E^0_A = \{\{g, ag\} : g \in G, a \in A\}, E^0_B = \{\{g, gb\} : g \in G, b \in B\}.
        \end{equation*}
    \item The squares are 
        \begin{equation*}
            F^0 = \{\{g, ag, gb, agb\} : g \in G, a \in A, b \in B\}.
        \end{equation*}
\end{itemize}

The main criterion for choosing $G$, $A$, and $B$ is to ensure the Cayley graphs $Cay(G, A)$ and $Cay(G, B)$ are good expanders, and in particular are Ramanujan. 
Besides this, for simplicity we further assume two technical conditions as in \cite{dinur2021}: that $|A| = |B| = \Delta$, and the so-called \textit{total no-conjugacy} condition
\begin{equation*}
    \forall a \in A, b \in B, g \in G, ag \ne gb.
\end{equation*}
The total no-conjugacy condition ensures squares are non-degenerate (contain exactly $4$ distinct vertices), and that each vertex is incident to exactly $k^2$ squares \cite[Claim 3.7]{dinur2021}. Leveraging classical results of Morgenstern \cite{morgenstern1994existence} and Lubotzky, Samuels, and Vishne \cite{lubotzky2005explicit}, \cite{dinur2021} show that explicit families of left-right Cayley complexes exist for infinitely many degrees.
\begin{theorem}[{\cite[Claim 6.7]{dinur2021}}]\label{thm:cayley}
    There exists an infinite sequence of degrees $\Delta = q + 1$ (where $q$ is an odd prime power) such that for each fixed $\Delta$ there exists an explicit infinite family of left-right Cayley complexes with $G_i = \textnormal{PSL}_2(q^i)$ and generator sets $A_i$ and $B_i$ such that $|A_i|=|B_i|=\Delta$, $Cay(G_i, A_i)$ and $Cay(G_i, B_i)$ are Ramanujan, and $A_i$,$B_i$ satisfy the total no-conjugacy condition.
\end{theorem}
As in \cite{leverrier2022quantum}, we will use the \textit{double cover} of the left-right Cayley complex, defined as:

\begin{itemize}
    \item The vertices are $V = V_0 \cup V_1$ where $V_0 = G \times \{0\}$ and $V_1 = G \times \{1\}$.
    \item The $A$-edges are $E_A$ and the $B$-edges are $E_B$ where
        \begin{equation*}
            E_A = \{\{(g, 0), (ag, 1)\} : g \in G, a \in A\}, E_B = \{\{(g, 0), (gb, 1)\} : g \in G, b \in B\}.
        \end{equation*}
    \item The squares are 
        \begin{equation*}
            F = \{\{(g, 0), (ag, 1), (gb, 1), (agb, 0)\} : g \in G, a \in A, b \in B\}.
        \end{equation*}
\end{itemize}
Note that every square in the original left-right Cayley complex corresponds to two squares in the double cover, and therefore that the double cover has a total of $\frac{\Delta^2|G|}{2}$ squares. Since we will only use the double cover in our arguments, from now on the term ``square'' will \textit{always refer to these double-covered squares}, not the squares in the original Cayley complex.

Following \cite{leverrier2022quantum}'s notation, we will mainly think of the double-covered complex as represented by the following graphs. First, we'll define a graph that captures the vertices and edge-structure of the Cayley complex: $\cG^\cup = (V, E_A \cup E_B)$. Second, we'll define graphs\footnote{We note these may technically be multi-graphs as in \cite{leverrier2022quantum}, but this has no effect on our arguments.} $\cG_0^\square = (V_0, E_0^\square)$ and $\cG_1^\square = (V_1, E_1^\square)$ capturing squares in the double cover, where 
\[
E_i^\square = \{\{(g, i), (agb, i)\} : g \in G, a \in A, b \in B\}
\]
for $i\in\{0,1\}$. Notice that the edges in these graphs have a one-to-one correspondence with the double-covered squares, namely that $E_i^\square \cong F$ for $i = 0, 1$ through the following identifications:
\[
\{(g, 0), (agb, 0)\} \leftrightarrow \{(g, 0), (ag, 1), (gb, 1), (agb, 0)\}
\]
and 
\[
\{(g, 1), (agb, 1)\} \leftrightarrow \{(g, 1), (ag, 0), (gb, 0), (agb, 1)\}.
\]
These identifications will be particularly important in the proof of small-set (co)-boundary expansion as we move between the squares of our complex and their associated graph representations.

Finally, it will be important to observe that these graphs inherit the spectral properties of $Cay(G, A)$ and $Cay(G, B)$. Namely that when the latter are Ramanujan, $\cG^\cup, \cG_0^\square, \cG_1^\square$ are also very good expanders.
\begin{lemma}[{\cite[Lemma 4]{leverrier2022quantum}}] \label{lem:expander}
    If $Cay(G, A), Cay(G, B)$ are Ramanujan graphs, then
        $\lambda(\cG_0^\square) \le 4 \Delta$, $\lambda(\cG_1^\square) \le 4 \Delta$, and $\cG^\cup$ is the double cover of a $4\sqrt{\Delta}$-spectral expander.
\end{lemma}
We note this is not exactly the statement given in \cite{leverrier2022quantum}, but the proof is the same.

\subsection{Error Correcting Codes}

A classical $(n,k,d)$-error correcting (erasure) code is a method for encoding a string of $k$ classical bits into $n > k$ classical bits such that one can recover the original string even when up to $d-1$ bits of the encoded string are erased. More formally, we will consider the standard setting of \textit{linear codes}, where the encoded space is a linear subspace $\cC \subset \FF_2^n$. Here $n$ is the \textit{length} of the code, $k \coloneqq \dim(\cC)$ is its \textit{dimension}, and the minimum weight of any element (also called codeword) of $\cC$, $d \coloneqq \min_{c \in \cC}\{ \hnorm{c}\}$, is called its \textit{distance}.\footnote{We note that this is similar to the distance operator $d(\cdot, \cdot)$ used to define co-boundary expansion. Indeed the distance of a code $\cC$ is just $d(\emptyset,\cC)$. We will abuse notation slightly to match standard coding theory notation and write this as $d(\cC)$ throughout.} One can check that in a linear code of distance $d$, it is indeed possible to uniquely correct up to $d-1$ errors. Finally, the ratio $r \coloneqq \frac{k}{n}$ is called the \textit{rate} of the code, and measures the overhead from the original to encoded space. We will typically be interested in families of codes that have constant rate and linear distance.

One of the main reasons to use linear codes is that there are nice linear algebraic ways of describing the objects. In particular, the linear subspace (code) $\cC$ is typically described either by a \textit{parity-check} matrix, or a \textit{generator} matrix. In particular, one can always find a parity-check matrix $M: \FF_2^n \rightarrow \FF_2^{n-k}$ whose kernel is the code in question ($\cC := \ker M \subset \FF_2^n$), and likewise a generator matrix
$M': \FF_2^k \rightarrow \FF_2^n$ whose image gives the code ($\cC := \Ima M' \subset \FF_2^n$). When clear from context, we sometime abuse notation and write $\cC$ to mean the parity check matrix of $\cC$.

\subsection{Tanner Codes}
The \textit{Tanner construction} (or \textit{tanner code}) \cite{tanner1981recursive} is a classical strategy in coding theory to build a linear code out of a `large' regular graph and a `small' \textit{local} code that sits on the neighborhood of each vertex. Crucially, when the underlying graph is an expander, it is often the case that the Tanner code inherits desirable properties from the small code. 

More formally, let $\cG = (V, E)$ be a $\Delta$-regular graph and $E(v)$ denote the set of edges incident to any $v \in V$. Assume an identification of $\FF_2^{E(v)}$ with $\FF_2^\Delta$ for each $v \in V$, which we call the \textit{local view} of $v$. Given a local code $C_0$ with length $\Delta$, the Tanner code $T(\cG, C_0) \subset \FF_2^E$ is given by
\begin{equation*}
    \{c \in \FF_2^E : \forall v \in V, c|_{E(v)} \in C_0\},
\end{equation*}
where $c|_{E(v)} \in \FF_2^{\Delta}$ is the vector formed by the values of $c$ on the local view of $v$.

It will be convenient for us to view the Tanner construction through its parity check matrix, which will make up the co-boundary operators of our chain complex. If our local code $C_0$ has parity check matrix $M_0$ and rate $r_0$, the parity check matrix of the Tanner code $T(\cG, C)$ is given by the composition:
    \begin{equation*}
        \FF_2^E \rightarrow \FF_2^{V \times \Delta} \rightarrow \FF_2^{V \times (1-r_0) \Delta}
    \end{equation*}
    where the first map copies the value on the edge to each local view of the vertices,
    and the second map applies $M_0$ to each local view independently for each vertex. We will sometimes refer to this parity check matrix as the \textit{Tanner map}.


\subsection{Robust Tensor Codes Against Puncture} \label{sec:robust-tensor-code}
The properties of our Tanner maps are highly dependent on the local code used to instantiate them. Following \cite{leverrier2022quantum}, we use a special type of local code called a \textit{tensor code}. We closely follow the discussion of these objects given in \cite{leverrier2022quantum}.

Recall that the generators of our left-right Cayley complex $A$ and $B$ have size $\Delta$. We will consider codes on $\FF_2^{A \times B}$ with tensor product structures. Namely, given two linear codes $C_A \subset \FF_2^A, C_B \subset \FF_2^B$, we define the \textit{tensor code} $C_A \otimes C_B$ to be the set of $\Delta \times \Delta$ matrices $M$ where each column vector $(M_{ab})_{a \in A}$ belongs to $C_A$ and each row vector $(M_{ab})_{b \in B}$ belongs to $C_B$. We define the \textit{dual tensor code} to be the sum $C_A \otimes \FF_2^B + \FF_2^A \otimes C_B$, where $C_A \otimes \FF_2^B$ are the $\Delta \times \Delta$ matrices whose columns belong to $C_A$, and $\FF_2^A \otimes C_B$ are the $\Delta \times \Delta$ matrices whose rows belong to $C_B$. The following claims about the dimension and distance of these codes are standard and easy to verify:
\begin{enumerate}
    \item $\dim(C_A \otimes C_B) = \dim(C_A) \dim(C_B)$
    \item $d(C_A \otimes C_B) = d(C_A) d(C_B)$
    \item $\dim(C_A \otimes \FF_2^B + \FF_2^A \otimes C_B) = \Delta \dim(C_A) + \Delta \dim(C_B) - \dim(C_A) \dim(C_B)$
    \item $d(C_A \otimes \FF_2^B + \FF_2^A \otimes C_B) = \min(d(C_A), d(C_B))$.
\end{enumerate}

To ensure our Tanner maps have the right properties, we will actually require our local tensor codes to have a stronger property called robustness. One can think of robustness as a generalization of distance of usual linear codes to the context of tensor codes, or as we will soon see, as a sort of robust testability property.
\begin{definition} [Robust {\cite[Definition 5]{leverrier2022quantum}}]
    Let $C_A \subset \FF_2^A, C_B \subset \FF_2^B$ be codes of length $\Delta$ of distance $d_A$ and $d_B$ respectively.
    We say the dual tensor code $C = C_A \otimes \FF_2^B + \FF_2^A \otimes C_B$ is $w$-robust if for every codeword $c \in C$ with Hamming weight $|c| < w$, there exist $A' \subset A, B' \subset B, |A'| \le |c|/d_B, |B'| \le |c|/d_A$, such that $c_{ab} = 0$ for any $a \not \in A'$ and $b \not \in B'$.
\end{definition}

Leverrier and Z\'emor \cite{leverrier2022quantum} prove that robust tensor codes satisfy a useful small-set robust testability property.

\begin{lemma} [\cite{leverrier2022quantum}, Proposition 6] \label{lem:small-set-robust}
  Let $C_A \subset \FF_2^A, C_B \subset \FF_2^B$ be codes of length $\Delta$ of distance $d_A$ and $d_B$ respectively.
  If the dual tensor code $C = C_A \otimes \FF_2^B + \FF_2^A \otimes C_B$ is $w$-robust with $w \le d_A d_B/2$, then any word $x$ close to both the column and row code is also close to the tensor code. More explicitly, if $d(x, C_A \otimes \FF_2^B) + d(x, \FF_2^A \otimes C_B) < w$ then:
  \[
    d(x, C_A \otimes C_B) \le \frac{3}{2} \left( d(x, C_A \otimes \FF_2^B) + d(x, \FF_2^A \otimes C_B) \right).
  \]
\end{lemma}
In fact, \cite{leverrier2022quantum} need a slightly stronger condition than just robustness of the code: it needs to remain robust even after the removal of a small set of rows and columns.
Conceptually, this is similar to the idea of smooth codes \cite{dinur2006robust} where the code maintains nice properties even after the removal of a small number of variables or checks. Given a code $C_A \subset \FF_2^A$ and $A' \subset A$, let $C_{A'} \subset \FF_2^{A'}$ denote the \textit{puncture code} which is the restriction of all codewords in $C_A$ to the coordinates in $A'$ (more precisely, $C_{A'} = \{(c_a)_{a \in A'} : (c_a)_{a \in A} \in C_A\}$).

\begin{definition} [Robust against puncture {\cite[Definition 7]{leverrier2022quantum}}]
    Given linear codes $C_A \subset \FF_2^A, C_B \subset \FF_2^B$, we say the dual tensor code $C_A \otimes \FF_2^B + \FF_2^A \otimes C_B$ is $w$-robust with $p$-resistance to puncture if for any $w'\leq p$ and $A' \subset A$ and $B' \subset B$ such that $|A'| = |B'| = \Delta-w'$, the dual tensor code $C_{A'} \otimes \FF_2^{B'} + \FF_2^{A'} \otimes C_{B'}$ is $w$-robust. 
\end{definition}
Extending prior work of \cite{panteleev2021asymptotically}, \cite{leverrier2022quantum} show random tensor codes are robust against puncture.

\begin{theorem} [{\cite[Theorem 8]{leverrier2022quantum}}] \label{thm:exist-robust-code}
    Let $0 < r_A < 1$ and $0 < r_B < 1$. Let $0 < \eps < 1/2$ and $1/2 + \eps < \gamma <1$. Let $C_A$ be a random code obtained from a random uniform $r_A \Delta \times \Delta$ generator matrix, and let $C_B$ be a random code obtained from a random uniform $(1-r_B) \Delta \times \Delta$ parity-check matrix.
    With probability tending to $1$ when $\Delta$ goes to
    infinity, the dual tensor code
    \begin{equation*}
        C_A \otimes \FF_2^B + \FF_2^A \otimes C_B
    \end{equation*}
    is $\Delta^{3/2-\eps}$-robust with $\Delta^\gamma$-resistance to puncturing.
\end{theorem}

Because the dual of a random code is again a random code, this implies both $C_A \otimes \FF_2^B + \FF_2^A \otimes C_B$ and $C_A^\perp \otimes \FF_2^B + \FF_2^A \otimes C_B^\perp$ are robust against puncture with high probability.
\begin{corollary}[{\cite[Theorem 17]{leverrier2022quantum}}]\label{cor:exist-robust-code-and-dual-robust-code}
    Fix $r \in(0,1/2)$, $\eps \in (0,1/2)$, $\gamma\in (1/2+\eps,1)$ and $\delta>0$ satisfying $-\delta\log \delta - (1-\delta)\log(1-\delta) < r$. When $k$ is large enough, there exist codes $C_A$ and $C_B$ of length $\Delta$ such that
    \begin{enumerate}
        \item $\dim C_A = \floor{r\Delta}$ and $\dim C_B = \Delta-\dim C_A$
        \item The distances of $C_A, C_B, C_A^\perp, C_B^\perp$ are all
        at least $\delta \Delta$
        \item Both dual tensor codes $C_0^\perp = (C_A\otimes C_B)^\perp$ and
        $C_1^\perp=(C_A^\perp \otimes C_B^\perp)^\perp$ are $\Delta^{3/2-\eps}$-robust
        with $\Delta^\gamma$-resistance to puncturing
        \item $C_A,C_B,C_A^\perp$, and $C_B^\perp$ have generator matrices where every row and column have at least two ones.
    \end{enumerate}
\end{corollary}
We note that this is not exactly the statement of \cite[Theorem 17]{leverrier2022quantum}, who prove the first three conditions occur with probability going to $1$ as $\Delta$ becomes large when $C_A$ and $C_B$ are generated as in \Cref{thm:exist-robust-code}. The fourth item is not included in \cite{leverrier2022quantum}, but also occurs under this distribution with high probability by fairly standard arguments. We give the proof in the appendix for completeness.
\section{Constructing Small-Set HDX}\label{sec:construction}
We are finally ready to construct a family of 3-term chain complexes with small-set boundary and co-boundary expansion. 

\begin{theorem}\label{thm:HDX} 
There exists an explicit infinite family of chain complexes $\{X_i\}$ and constants $d \in \mathbb{N}$ and $\rho_1, \rho_2 \in (0,1)$ such that each $X_i$ satisfies:
\begin{enumerate}
    \item $X_i$ has maximum degree $d$ and minimum degree at least $3$
    \item $X_i$ has non-trivial co-homology $H^1$
    \item $X_i$ is a $(\rho_1,\rho_2)$-small-set HDX.
\end{enumerate}
\end{theorem} 
Combined with \Cref{thm:HDX-to-XOR} which transforms SS-HDX into hard instances of 3-XOR, this completes the proof of our main theorem.
\begin{proof}[Proof of \Cref{intro:thm-main}]
    \Cref{thm:HDX-to-XOR} gives the desired explicit family of 3-XOR instances as long as it is provided an explicit family of chain complexes with bounded maximum degree, minimum degree at least $3$, non-trivial co-homology, and which are $(\rho_1,\rho_2)$-small-set boundary expanders with $\mu$-co-systolic distance for some set of constants $\mu,\rho_1,\rho_2 \in (0,1)$. Since any $(\rho_1,\rho_2)$-small-set co-boundary expander has $\rho_1$-co-systolic distance (\Cref{lem:ss-boundary-to-systolic-distance}), \Cref{thm:HDX} provides an explicit family of chain complexes matching these conditions with $\mu=\rho_1$.
\end{proof}

As discussed, \Cref{thm:HDX} is proved via Leverrier and Z\'emor's \cite{leverrier2022quantum} recent construction of good qLPDC codes. They show the associated 3-term chain complex has linear systolic and co-systolic distance. Our contribution is to observe that the same construction actually satisfies the stronger small-set boundary and co-boundary expansion conditions. We note that while we only show this property for Leverrier and Z\'emor's \cite{leverrier2022quantum} simplified construction, similar arguments likely hold for Panteleev and Kalachev's \cite{panteleev2021asymptotically} original good qLDPC codes as well.

\paragraph{Construction:} We first describe Leverrier and Z\'emor's construction, which is based upon Tanner maps (parity-check matrices of Tanner codes). To start, we'll first need to describe the underlying graphs and local codes of these maps. Recall the explicit family of left-right Cayley complexes promised by \Cref{thm:cayley} and for any fixed complex $X_i$ in the family let the group $G=G_i$ and generator sets $A=A_i, B=B_i$ be as in the theorem. The graphs underlying our Tanner maps will be the `square graphs' $\cG_0^\square$ and $\cG_1^\square$, which we recall have
\begin{itemize}
    \item Vertices $V_i = G \times \{i\}$,
    \item Edges $E_i^\square = \{\{(g, i), (agb, i)\} : g \in G, a \in A, b \in B\}$
\end{itemize}
for $i \in \{0,1\}$ respectively. It bears repeating that edges in these graphs are in one-to-one correspondence with squares of the double covered Cayley complex via the following identifications:
\[
\{(g, i), (agb, i)\} \leftrightarrow \{(g, i), (ag, 1-i), (gb, 1-i), (agb, i)\}.
\]
We will frequently refer to edges in $\cG_i^\square$ as squares due to this connection.

Since the square graphs $\cG_i^\square$ are $\Delta^2$-regular, we can define a Tanner map by combining them with any length $\Delta^2$ local code. This role will be played by the robust dual tensor codes promised by \Cref{cor:exist-robust-code-and-dual-robust-code}. Namely, letting $C_A: \FF_2^\Delta \rightarrow \FF_2^{r\Delta}$ and $C_B: \FF_2^{\Delta} \rightarrow \FF_2^{(1-r)\Delta}$ be as in \Cref{cor:exist-robust-code-and-dual-robust-code} for some choice of $r,\varepsilon$, and $\gamma$, our local codes will be $C_0^\perp$ and $C_1^\perp$ where $C_0 = C_A \otimes C_B$ and $C_1 = C_A^\perp \otimes C_B^\perp$.\footnote{Note we are assuming for simplicity that $r\Delta$ and $(1-r)\Delta$ are integer valued, but these can be replaced with $\floor{r\Delta}$ and $\Delta-\floor{r\Delta}$ without substantially affecting the proof (see \cite{leverrier2022quantum}).}

Combining these graphs and local codes gives the Tanner maps $\cC_0 = T(\cG_0^\square, C_0^\perp) : \FF_2^n \rightarrow \FF_2^{m}$ and $\cC_1 = T(\cG_1^\square, C_1^\perp) : \FF_2^n \rightarrow \FF_2^{m}$, where $n = |F| = \Delta^2 |G|/2$ is the number of squares, and $m = r(1-r)\Delta^2|V_0| = r(1-r)\Delta^2|G|$ comes from the fact that the both dual tensor codes have dimension $(1-r(1-r))\Delta^2$.
Associating the edges of $\cG_i^\square$ with squares in the discussed manner, one can check that $\cC_1 \cC_0^T = 0$ (see \cite[Section 4.1]{leverrier2022quantum}) and therefore that these maps define a chain complex:
\begin{equation}\label{eq:chain-complex-construction}
X: \FF_2^{m} \xrightarrow{\delta_0 := \cC_0^T} \FF_2^n \xrightarrow{\delta_1 := \cC_1} \FF_2^{m}.
\end{equation}
Moreover, this process gives an explicit family of chain complexes $\{X_i\}$ by choosing $G_i$, $A_i$, and $B_i$ as in the explicit family of left-right Cayley complexes promised by \Cref{thm:cayley}, and computing $C_A, C_B$ with the desired properties by brute force search over all pairs of length $\Delta$ codes $C_A,C_B$ of dimensions $r\Delta$ and $(1-r)\Delta$ respectively.\footnote{Note that since $\Delta$ is a constant with respect to our infinite family, brute force search only requires $O(1)$ time here.}



This completes the construction. We now move to showing that $X$ has the three desired properties: bounded-degree, non-trivial co-homology, and small-set (co)-boundary expansion.

\paragraph{$X$ has (upper) bounded-degree:} By definition $X$ is bounded-degree if and only if the parity-check matrices of our two Tanner codes have a bounded number of ones in every row and column. By the nature of the Tanner code construction the support of any row or column is at most twice the degree of the underlying graph. Since our graphs are of degree $\Delta^2$ (a constant with respect to the family), the resulting complex is bounded-degree as desired.

\paragraph{$X$ has (lower) bounded-degree:} Recall we are promised that $C_A,C_B,C_A^\perp$, and $C_B^\perp$ have generator matrices where every row and column have at least two ones. This implies that the tensor codes $C_A \otimes C_B$ and $C_A^\perp \otimes C_B^\perp$ can be taken to have generator matrices with at least four ones in each row and column. Since these correspond to the parity check matrices of $C_0^\perp = (C_A \otimes C_B)^\perp$ and $C_1^\perp = (C_A^\perp \otimes C_B^\perp)^\perp$ respectively, it can be easily checked that the parity check matrices of the associated Tanner codes $T(\cG_0^\square, C_0^\perp)$ and $T(\cG_1^\square, C_1^\perp)$ also have at least four ones in every row and column.

\paragraph{$H^1$ is non-trivial:} This follows immediately from dimensionality arguments. In particular, notice that $\dim Z^1 \geq n-m$, whereas $\dim B^1 \leq m$. As a result we have $\dim H^1 \ge n - 2m = (1/2 - 2 r(1-r)) \Delta^2 |G|$ which is $> 0$ whenever $r \ne 1/2$.

\paragraph{$X$ is a small-set (co)-boundary expander:}

It is left to show our complexes are small-set (co)-boundary expanders. In what follows we show the co-boundary expansion case. Since the construction is symmetric, a similar proof gives small-set boundary expansion. For convenience, we first re-formulate the problem as the following technical theorem. Note that this is the analog of \cite[Theorem 1]{leverrier2022quantum} where co-systolic distance is replaced with small-set co-boundary expansion. We follow their notation when possible for consistency.

\begin{theorem}\label{thm:SS-boundary-technical}
    Fix $\eps \in (0,1/2)$, $\gamma\in (1/2+\eps,1)$ and $\delta>0$. 
    For any fixed large enough $\Delta$, if the linear codes $C_A$ and
    $C_B$ have minimum distance at least $\delta \Delta$ and 
    if the dual tensor code $C_A\otimes\FF_2^B +\FF_2^A\otimes C_B=C_1^\perp$ 
    is $w$-robust with $p$-resistance to puncturing for $w = \Delta^{3/2-\eps/2}$ and
    $p=\Delta^{\gamma}$,\footnote{We note that the value of $w$ here is slightly different than in \cite{leverrier2022quantum}. This corrects a small error in the application of robust testability (\Cref{lem:small-set-robust}) in the original work.} then the chain complex in \Cref{eq:chain-complex-construction}:
\[
X: \FF_2^{m} \xrightarrow{\delta_0 := \cC_0^T} \FF_2^n \xrightarrow{\delta_1 := \cC_1} \FF_2^{m},
\]
satisfies the following isoperimetric inequality for small, minimal chains: 
    \[
    \forall x \in \mathbb{F}_2^n ~\text{s.t. $x$ is minimal and}~ \hnorm{x} \leq \rho_1 n:  \hnorm{\delta_1 x} \ge \rho_2\hnorm{x},
    \]
where $\rho_1 = \frac{\delta}{6 \Delta^{3/2 + \eps}}$, $\rho_2 = \frac{56}{\Delta^{3 - 2 \eps}}$.
\end{theorem}
Recall that this isoperimetric condition is equivalent to $(\rho_1,\rho_2)$-small-set co-boundary expansion (\Cref{lemma:SS-LTC}), so this indeed proves the desired property. The proof of \Cref{thm:SS-boundary-technical} closely follows the analogous proof in \cite{leverrier2022quantum} for systolic distance. The main difference is that we must track an additional set of elements consisting of vertices in $\cG^{\square}_0$ corresponding to violated constraints. Since \cite{leverrier2022quantum} only need to consider $x \in \mathbb{F}^n_2$ that are true codewords, this is not a relevant consideration in their result. We note that throughout we set our coefficients to match those in \cite{leverrier2022quantum} for ease of comparison.

\begin{proof}[Proof of \Cref{thm:SS-boundary-technical}]
    We assume $x \ne 0$, as the theorem holds trivially otherwise. We proceed by contradiction. Assuming $\hnorm{\delta_1 x} < \rho_2 \hnorm{x}$, we will show there exists $y \in B_1$ such that $\hnorm{x+y} < \hnorm{x}$, contradicting minimality of $x$.

    We first lay out some relevant notation. Thinking of $x$ as a subset of $E_1^\square$ (the edge set of $\cG^\square_1$), we will consider the edge-induced subgraph $\cG^\square_{1,x} \subset \cG^\square_1$ and denote its vertex set by  $S \subset V_1$. Recall that each vertex $(g,1) \in V_1$ has a corresponding local view made up of $|A|\cdot|B|$ incident squares, which we'll denote by: 
    \[
    \ell(g,1) \coloneqq \left\{\{(g,1), (ag, 0), (gb, 0), (agb, 1)\}: a \in A, b \in B\right\}.
    \]
    Thinking of $x$ now as a set of squares, let $x|_{\ell(g,1)} \in \mathbb{F}_2^{A \times B}$ denote the restriction of $x$ to the local view of $(g,1)$, and recall that $x$ is a co-cycle exactly when these local views correspond to codewords in $C_1^\perp = C_A \otimes \FF_2^B + \FF_2^A \otimes C_B$. 
    
    Since $x$ is arbitrary in our setting (unlike \cite{leverrier2022quantum} who only consider co-cycles) we will partition the vertices of our induced subgraph into three parts: $S = S_v \cup S_n \cup S_e$. First, let $S_v \subset S$ denote the set of \textit{violated vertices} $(g,1) \in V_1$ whose local views $x|_{\ell(g,1)}$ do \textit{not} form codewords in $C_1^{\perp}$. Following \cite{leverrier2022quantum}, we split the remaining vertices in $S \setminus S_v$ into two parts based upon their degree in the induced subgraph $\cG^\square_{1,x}$: the \textit{normal vertices} $S_n$ with degree less than $w_2 := \Delta^{3/2-\eps}$, and the \textit{exceptional vertices} $S_e$ with degree at least $w_2$. The intuition behind this strategy is that because $C_1^\perp$ is $(w>w_2)$-robust, the codewords associated to vertices in $S_n$ have particularly nice structure: they are zero outside of a small set of at most $w_2/(\delta \Delta)$ rows and columns. This implies that any column (respectively row) is close to a codeword in $C_A$ (respectively $C_B$)
    which will eventually help us apply small-set robust testability (\Cref{lem:small-set-robust}) to prove $x$ is close to a co-boundary (and is therefore non-minimal). 
    
    To find such a co-boundary, we'll first need to look to the other side of the complex. Broadly speaking, the idea (which is the same as in \cite{leverrier2022quantum}) is to find a vertex $v \subset V_0$ whose local view shares many (heavy) rows and columns with local views of vertices in $S_n$. One can then apply robustness to see that the value of $x$ on this local view is close to a codeword in $C_A \otimes C_B$ which can easily be translated to the desired co-boundary. 
    
    More formally, let $E_x \subset \cG^\cup$ denote the set of edges incident to the squares in $x$,\footnote{In particular for any square  $\{(g,0), (ag, 1), (gb, 1), (agb, 0)\} \in x$, add its four edges $\{(g,0), (ag, 1)\}$, $\{(g,0), (bg, 1)\}$, $\{(agb,0), (ag, 1)\}$, and $\{(agb,0), (gb, 1)\}$ to $E_x$.}
    and call an edge \textit{heavy} if it is incident to at least $\delta \Delta - \Delta^{1/2-\eps}/\delta$ squares in $x$. We will consider the set of vertices $T \subset V_0$ which are adjacent to $S_n \subset V_1$ through a heavy edge in the graph $\cG^\cup$. Given $v \in T$, note that every heavy edge with an element in $S_n$ corresponds to a row or column that is shared in their local view (and is therefore close to a codeword of $C_B$ or $C_A$ respectively). The goal is therefore to show that there exists a vertex in $T$ that is adjacent to many elements in $S_n$ through heavy edges, while simultaneously adjacent to few `bad' vertices in $S_e$ and $S_v$. This will allow us to apply robustness against puncture to find a co-boundary that reduces the weight of $x$. We formalize these statements below in the following two claims.

    \begin{claim} [Modification of {\cite[Claim 13]{leverrier2022quantum}}] \label{claim:many-heavy-little-ev}
        There exist $h_1 \geq \Omega(\Delta)$, $d_1 \leq O(\Delta^{1/2+\eps})$, and $v\in T$ such that $v$ is incident to at least $h_1$ heavy edges and adjacent to at most $d_1$ vertices of $S_e \cup S_v$.
    \end{claim}

    \begin{claim} [Summary of paragraph following {\cite[Claim 13]{leverrier2022quantum}}] \label{claim:exist-flip}
    For all sufficiently large\footnote{Here we mean in terms of $r,\varepsilon$, and $\gamma$, so $\Delta$ remains constant with respect to the infinite family.} $\Delta$, if there exists a vertex $v \in V_0$ incident to $h_1 \geq \Omega(\Delta)$ heavy edges and at most $d_1 \leq O(\Delta^{1/2+\eps})$ vertices of $S_e \cup S_v$, then we can find a vector $y \in B_1$ such that $\hnorm{x+y} < \hnorm{x}$.
    \end{claim}
Together, \Cref{claim:many-heavy-little-ev} and \Cref{claim:exist-flip} complete the proof of \Cref{thm:SS-boundary-technical}, as they promise the existence of some $y \in B_1$ such that $\hnorm{x+y} < \hnorm{x}$, violating minimality of $x$. Thus it is left to prove the claims. While \Cref{claim:exist-flip} follows largely from arguments in \cite{leverrier2022quantum}, it is helpful to present first to motivate the more technical proof of \Cref{claim:many-heavy-little-ev}.
    \begin{proof} [Proof of \Cref{claim:exist-flip}]
        Recall we are given an element $v \in V_0$ which is incident to at least $h_1 \geq \Omega(\Delta)$ heavy edges and adjacent to at most $d_1 \leq O(\Delta^{1/2+\varepsilon})$ vertices in $S_e \cup S_v$. We consider the local view of $x$ around $v$ (considered as an element of $\cG^\square_0$), denoted $x_v \in \FF_2^{A \times B}$ here for notational simplicity.
        Because at most $d_1$ vertices adjacent to $v$ in $\cG^{\cup}$ are exceptional or violated (as considered in $\cG^{\square}_{1,x}$), one can find $A' \subset A, B' \subset B$ with $|A'| = |B'| \geq \Delta - d_1$, such that $A'$ and $B'$ are indexed by either normal vertices, or vertices not in $S$. Furthermore, since $d_1 \leq \Delta^\gamma$ for large enough $\Delta$, we also have by robustness to puncture that the restricted dual tensor code $(C_1^{\perp})' \coloneqq C_{A'} \otimes \mathbb{F}^{B'}_2 + \mathbb{F}^{A'}_2 \otimes C_{B'}$ is $w$-robust.

        Let $x_v'$ be the restriction of $x_v$ in $A' \times B'$. Recall each column (row) of the local view of a normal vertex is at most $w_2/(\delta \Delta)$ away from a codeword by $w$-robustness. Then since each column (row) of $x_v'$ is indexed by either a normal vertex or a vertex whose local view is all zero (i.e. not in $S$), every column (respectively row) of $x_v'$ is at most $w_2/(\delta \Delta)=\Delta^{1/2-\eps}/\delta$ away from a codeword in $C_{A'}$ (respectively $C_{B'}$). Since there are at most $\Delta$ rows and columns, this means that $x_v'$ is at most $\Delta^{3/2-\varepsilon}/\delta$ away from either $C_{A'} \otimes \FF_2^{B'}$ or $\FF_2^{A'} \otimes C_{B'}$, and moreover that:
        \[
        d(x_v', C_{A'} \otimes \FF_2^{B'}) + d(x_v', \FF_2^{A'}\otimes C_{B'}) \leq 2\Delta^{3/2-\varepsilon}/\delta \leq w
        \]
        for sufficiently large $\Delta$. Because $(C_1^{\perp})'$ is $w$-robust, we can apply small-set robust testability (\Cref{lem:small-set-robust}) to infer that $x_v'$ is close to some codeword $c' \in C_{A'} \otimes C_{B'}$:
        \begin{equation*}
            d(x_v', c') \le \frac{3}{2} \left(d(x_v', C_{A'} \otimes \FF_2^{B'}) + d(x_v', \FF_2^{A'}\otimes C_{B'}) \right) \le 3 \frac{\Delta^{3/2-\eps}}{\delta}.
        \end{equation*}
        Finally, since the total number of punctured rows and columns is less than the code distance for large enough $\Delta$, we can extend $c'$ uniquely to a codeword $c \in C_A \otimes C_B$. Taking into account the rows and columns added in this process, the distance from $x_v$ to $c$ then becomes at most $d(x_v, c) \leq d(x_v',c') + 2d_1\Delta \leq O(\Delta^{3/2 + \eps}) < o(\Delta^2)$ since $\varepsilon < 1/2$.

        On the other hand, because $v$ is incident to $\Omega(\Delta)$ heavy edges, the weight $|x_v| = \Theta(\Delta^2)$. Thus for large enough $\Delta$, it must be the case that flipping $c$ strictly reduces the weight of $x$. More precisely, set $y$ to be $c$ on the local view $x_v$ and $0$ elsewhere, then we have $\hnorm{x+y} < \hnorm{x}$. Since $c \in C_A \otimes C_B = C_0$, $y \in B^1$ is indeed a co-boundary which completes the proof.
    \end{proof}
    The only thing left is to show that our main technical claim actually holds, the existence of a vertex with many heavy edges that is adjacent to few violated or exceptional vertices. The proof technique is similar to that of \cite[Claim 13]{leverrier2022quantum}, and mostly boils down to proving that $S_v$ and $S_e$ are small compared to $S_n$.
    \begin{proof} [Proof of \Cref{claim:many-heavy-little-ev}]
    We split the proof into the following three claims. First, we claim $T$ is non-empty.
    \begin{claim}\label{claim:nonempty}
            $|T| > 0$.
    \end{claim}
    With this in mind, let $\alpha,\beta=\Theta(1)$ be constants to be set later in the proof. Following \cite{leverrier2022quantum}, we claim that a reasonable fraction of $T$ is incident to many heavy edges:
        \begin{claim} [{\cite[Claim 12]{leverrier2022quantum}}] \label{claim:many-heavy}
            At least an $\alpha/2$ fraction of vertices in $T$ are incident to at least $h_1 = \alpha \Delta$ heavy edges,
        \end{claim}
        \noindent and further that at most some smaller fraction is adjacent to greater than $d_1$ violated and exceptional vertices:
        \begin{claim} [{\cite[Paragraph between Claim 4.10 and Claim 4.11]{leverrier2022quantum}}] \label{claim:little-ev2}
            At most an $\alpha/4$ fraction of vertices in $T$ are incident to more than $d_1 = \frac{4\beta}{\alpha} \Delta^{1/2 + \eps}$ vertices of $S_e \cup S_v$.
        \end{claim}
         Combining these claims implies at least an $\alpha/4$ fraction of vertices satisfy the requirements of \Cref{claim:many-heavy-little-ev}. Since $T$ is non-empty, this must apply to at least one $v \in T$ which gives the desired result.

        The key to proving all three claims lies in showing that the number of vertices in $S_e \cup S_v$ is small compared to $S_n$. We will show that $S_v$ can be upper bounded by taking $\rho_2$ sufficiently small, and $S_e$ can be upper bounded by the expander mixing lemma as in \cite{leverrier2022quantum}. 
        \begin{lemma} [Modification of {\cite[Claim 6]{leverrier2022quantum}}] \label{claim:little-ev}
            The number of exceptional and violated vertices is at most
            \begin{equation}
                |S_e \cup S_v| \le \frac{64}{\Delta^{1-2\eps}} |S|.
            \end{equation}
            On the other hand, the number of normal vertices is at least
            \begin{equation}
                |S_n| \ge (1-\frac{64}{\Delta^{1-2\eps}}) |S|.
            \end{equation}
        \end{lemma}
        \begin{proof}
        The latter fact follows immediately from the former and recalling that $S_e$, $S_v$, and $S_n$ partition $S$. We now show $|S_v|$ is small.
            Note that by assumption we have that
            \[
            |S_v| \leq \hnorm{\delta_1 x} < \rho_2 \hnorm{x},
            \]
            since  $\delta_1$ is the parity-check matrix of $\cC_1$ and every violated vertex corresponds to at least one violated constraint in $\cC_1$.
            Because $V_1$ has degree $\Delta^2$ in $\cG_1^\square$ (i.e.\ each vertex sits in $\Delta^2$ squares), we also have $\hnorm{x} \le \Delta^2 |S|/2$.
            Altogether this gives 
            \[|S_v| < \rho_2 \Delta^2 |S|/2 = 28|S|/\Delta^{1-2\eps}
            \]
            for our choice of $\rho_2$.

            Now we show $|S_e|$ is small.
            The degree of each non-violated vertex is at least $\delta \Delta$ because the local view corresponds to a non-zero codeword in $C_A \otimes \FF_2^B + \FF_2^A \otimes C_B$. 
            This implies $|S_n| + |S_e| \le \frac{2|x|}{\delta \Delta}$.
            Combining this with our bound on $|S_v|$ gives
            \begin{equation} \label{eq:bound-S}
                |S| = |S_n|+|S_e|+|S_v| \leq (\rho_2 + \frac{2}{\delta \Delta}) |x| \le \frac{4}{\delta \Delta} |x|
            \end{equation}
            where the second inequality holds for large enough $\Delta$ (recalling that $\rho_2 = O(\Delta^{-3+2\eps})$). Applying the expander mixing lemma to $E(S_e, S)$, we then obtain
            \begin{align*}
                |E(S_e,S) | &\le \frac{\Delta^2}{|V_1|}|S_e||S| + 4 \Delta \sqrt{|S_e| |S|} \\
                    &\le \frac{4 \Delta}{\delta |V_1|}|x| |S_e| + 4 \Delta \sqrt{|S_e| |S|} \\
                    &= \frac{1}{3}\Delta^{3/2-\eps}|S_e| + 4 \Delta \sqrt{|S_e| |S|}
            \end{align*}
            where we have used the assumption that $|x| \leq \frac{\delta n}{6 \Delta^{3/2+\eps}}$ and the fact that $|V_1|=|G| = 2n/\Delta^2$. On the other hand, by definition of exceptional vertices we have that $|E(S_e, S)| \ge \Delta^{3/2-\eps}|S_e|$. Combining the inequalities we obtain $|S_e| \le 36 |S|/\Delta^{1-2\eps}$, and plugging in our bound on $|S_v|$ then gives $|S_e \cup S_v| \le 64 |S|/\Delta^{1-2\eps}$ as desired.
        \end{proof}

        Finally, we prove \Cref{claim:nonempty}, \Cref{claim:many-heavy}, and \Cref{claim:little-ev2}, completing the result. The latter two follow essentially as in \cite{leverrier2022quantum} (replacing $S_e$ with $S_e \cup S_v$), but we give the proofs here for completeness.
        \begin{proof}[Proof of \Cref{claim:nonempty}]
        We wish to prove $T$ is non-empty. First, recall that since $x \ne 0$ by assumption, $|S| > 0$. By \Cref{claim:little-ev}, we then have $|S_n| > 0$ as well. We now argue that every vertex in $S_n$ is incident to at least one heavy edge. Since $S_n$ is non-empty, this implies $T$ is non-empty as desired.
        
        To see each vertex in $S_n$ has a heavy edge, recall the local view of each normal vertex is a codeword in $C_A \otimes \FF_2^B + \FF_2^A \otimes C_B$ with weight less than $w=\Delta^{3/2-\eps}$. Because the dual tensor code is $w$-robust, each column (respectively row) is within $\Delta^{1/2-\eps}/\delta$ of a codeword in $C_A$ (respectively $C_B$). Since these codes all have distance at least $\delta \Delta$, there must be a row or column with at least $\delta \Delta - \Delta^{1/2-\eps}/\delta$ ones which exactly corresponds to a heavy edge. We note this fact also implies the total number of heavy edges is at least $|S_n|$, which will be useful later on.
        \end{proof}
        \begin{proof} [Proof of \Cref{claim:many-heavy}]
        Now that we have confirmed the existence of $T$, we want to show it is incident to many heavy edges. To do so, we'll argue that $T$ is small compared to the number of heavy edges. 
        
        To start, we show that $|T| \le \frac{64}{\delta^2 \Delta} |S|$. The proof is the same as \cite[Claim 11]{leverrier2022quantum}, but we give it here for completeness. First, note that by the expander mixing lemma on $\cG^\cup$ (which is the double cover of a $4 \sqrt{\Delta}$-spectral expander) we have:
            \begin{align*}
            |E(S, T)| &\le \frac{2\Delta}{|G|}|S||T| + 4 \sqrt{\Delta} \sqrt{|S||T|}\\
            &\le \frac{2\Delta^{1/2-\varepsilon}}{3}|T| + 4 \sqrt{\Delta} \sqrt{|S||T|}
            \end{align*}
            where as in \Cref{claim:little-ev} we have again used the fact that
            \[
                        \hnorm{S} \leq \frac{4}{\delta\Delta}\hnorm{x} \leq \frac{2}{3\Delta^{5/2+\varepsilon}}n = \frac{1}{3\Delta^{1/2+\varepsilon}}|G|.
            \]
            On the other hand, since each vertex $v \in T$ is incident to at least one heavy edge $e$ by definition, $v$ (and $e$) are contained in at least $\delta \Delta -\Delta^{1/2-\eps}/\delta$ squares in $x$. Since each of these contains an additional (unique) edge incident to $v$, we also have the following lower bound
            \[
            |E(S, T)| \ge (\delta \Delta -\Delta^{1/2-\eps}/\delta) |T|.
            \]
            Combining these inequalities one can check that $|T| \le \frac{64}{\delta^2 \Delta} |S|$ for large enough $\Delta$ as desired.

            With this in hand, recall from the proof of \Cref{claim:nonempty} that the total number of heavy edges in $E_x$ is at least $|S_n| \geq (1-\frac{64}{\Delta^{1-2\eps}}) |S|$ (where the inequality is given by \Cref{claim:little-ev}). Together, this implies the average number of heavy edges incident to a vertex in $T$ is at least:
            \begin{equation}
                \frac{|S_n|}{|T|} \ge \frac{\delta^2 \Delta}{64} \left(1-\frac{64}{\Delta^{1-2\eps}} \right) =: 2 \alpha \Delta.
            \end{equation}
            Finally given that the average degree is at least $2\alpha \Delta$, we want to show there is some fraction of vertices with degree $\ge \alpha \Delta$. This is immediate from recalling that the maximum degree of $\cG^{\cup}$ (and thus $T$) is $2 \Delta$, which implies at least an $\alpha/2$ fraction of vertices in $T$ are incident to at least $\alpha \Delta$ heavy edges as desired.
        \end{proof}
        \begin{proof} [Proof of \Cref{claim:little-ev2}]
            Finally, we want to show there are few edges between $T$ and $S_e \cup S_v$. This follows from the fact that both sets are small, and the underlying graph $\cG^\cup$ is the double cover of a $4\sqrt{\Delta}$-expander on $|G|=\frac{2n}{\Delta^2}$ vertices. In particular, combining the expander mixing lemma with our bounds from \Cref{claim:little-ev} gives:
            \begin{align*}
            E(S_e \cup S_v, T) &\le \frac{2\Delta}{|G|}|S_e \cup S_v||T| + 4 \sqrt{\Delta} \sqrt{|T||S_e \cup S_v|}\\
            &\leq \frac{128\Delta^{2\varepsilon}}{|G|}|S||T| + 32\Delta^{\varepsilon} \sqrt{|T||S|}.
            \end{align*}
            Recall that $\hnorm{S} \leq  \frac{1}{3\Delta^{1/2+\varepsilon}}|G|$. Further, since each normal vertex is adjacent to $T$ and the degree of $T$ is at most $2\Delta$, we have $(1-\frac{64}{\Delta^{1-2\eps}}) |S| \leq |S_n| \le 2 \Delta |T|$, and thus for large enough $\Delta$ that $|S|\le 4\Delta|T|$. Altogether we therefore have:
            \begin{align*}
            E(S_e \cup S_v, T) &\le  \frac{128}{3\Delta^{1/2-\varepsilon}}|T| + 64\Delta^{1/2+\varepsilon} |T| \leq \beta \Delta^{1/2 + \eps} |T|
            \end{align*}
            where $\beta = 64 + \frac{128}{3\Delta}$.
            As a result, at most an $\alpha/4$ fraction of vertices in $T$ are incident to more than $d_1 = \frac{4\beta}{\alpha} \Delta^{1/2 + \eps}$ vertices of $S_e \cup S_v$ as desired, which completes the proof of \Cref{claim:many-heavy-little-ev} and \Cref{thm:SS-boundary-technical} in turn.
        \end{proof}
    \end{proof}
\end{proof}

Putting everything together, we now prove the existence of an explicit family of SS-HDX.
\begin{proof}[Proof of \Cref{thm:HDX}]
    Fix any $r \in (0,1/2)$, $\varepsilon \in (0,1/2)$, $\gamma \in (1/2+\varepsilon,1)$, and $\delta \in (0,1)$ satisfying $-\delta\log \delta - (1-\delta)\log(1-\delta) < r$, and let $\Delta=\Delta(r,\varepsilon/2,\gamma,\delta) \in \mathbb{N}$ be sufficiently large that the guarantees of \Cref{cor:exist-robust-code-and-dual-robust-code} and \Cref{thm:SS-boundary-technical} are met. Brute forcing over pairs of length $\Delta$ codes $C_A,C_B$ of dimensions $r\Delta$ and $(1-r)\Delta$ respectively, \Cref{cor:exist-robust-code-and-dual-robust-code} promises we can find in $O_{\Delta}(1)$ time codes $C_A,C_B$ such that:
    \begin{enumerate}
        \item $\dim C_A = \floor{r\Delta}$ and $\dim C_B = \Delta-\dim C_A$,
        \item The distances of $C_A, C_B, C_A^\perp, C_B^\perp$ are all
        at least $\delta \Delta$,
        \item Both dual tensor codes $C_0^\perp = (C_A\otimes C_B)^\perp$ and
        $C_1^\perp=(C_A^\perp \otimes C_B^\perp)^\perp$ are $\Delta^{3/2-\eps/2}$-robust
        with $\Delta^\gamma$-resistance to puncturing.
        \item $C_A,C_B,C_A^\perp$, and $C_B^\perp$ have generator matrices where every row and column have at least two ones.
    \end{enumerate}
Following the construction and the discussion earlier this section, the Tanner maps resulting from these codes and the explicit left-right Cayley complexes of \cite{dinur2021} give an explicit family of chain complexes with degree between $3$ and $2 \Delta^2$ and non-trivial co-homology. Furthermore each individual complex in the family satisfies the requirements of \Cref{thm:SS-boundary-technical} in both directions, so by symmetry the complexes are $(\rho_1,\rho_2)$-small-set HDX for $\rho_1 = \frac{\delta}{6 \Delta^{3/2 + \eps}}$ and $\rho_2 = \frac{56}{\Delta^{3 - 2 \eps}}$. This concludes the proof of \Cref{thm:HDX}.
\end{proof}

\begin{remark}\label{remark:local}
We note that the proof of \Cref{thm:SS-boundary-technical} actually gives a stronger guarantee than small-set (co)-boundary expansion. In particular, because the boundary $y$ that reduces the weight of $x$ is supported on a local view of a single vertex, the result actually gives an isoperimetric inequality for the broader class of small, locally minimal functions:
    \begin{equation*}
        \forall x \in \mathbb{F}_2^n ~\text{s.t. $x$ is locally minimal and}~ \hnorm{x} \leq \rho_1 n:  \hnorm{\delta_1 x} \ge \rho_2\hnorm{x},
    \end{equation*}
where $x$ is locally minimal if $\hnorm{x} \leq \hnorm{x+\delta_0(e_v)}$ for all basis vectors $e_v \in \mathbb{F}_2^{m}$.  As discussed in \Cref{sec:related-work}, this stronger isoperimetric inequality has seen prior use in the topological HDX literature \cite{kaufman2014ramanujan,evra2016bounded,kaufman2018cosystolic,kaufman2021unique} as well as in recent work on c3-LTCs \cite{lin2022c} and qLDPC codes \cite{lin2022good}.
\end{remark}

\section*{Acknowledgements}
The authors thank Noah Fleming, Sam Hopkins, and Russell Impagliazzo for helpful discussion on reductions within the Sum-of-Squares hierarchy, Amy Kanne for helpful discussions on qLDPC codes and \cite{leverrier2022quantum}, and Tali Kaufman for many fruitful discussions on high dimensional expansion. The authors also thank Sam Hopkins, Shachar Lovett, and Anthony Ostuni for helpful comments on an earlier version of the manuscript. 

\bibliographystyle{amsalpha}  
\bibliography{references} 
\newpage
\appendix
\section{Existence of Good Base Codes}
In this section we prove \Cref{cor:exist-robust-code-and-dual-robust-code}, the existence of base codes $C_A$ and $C_B$ with the properties needed for our SS-HDX construction in \Cref{thm:HDX}. We restate the result here for convenience.
\begin{corollary}
    Fix $r \in(0,1/2)$, $\eps \in (0,1/2)$, $\gamma\in (1/2+\eps,1)$ and $\delta>0$ satisfying $-\delta\log \delta - (1-\delta)\log(1-\delta) < r$. When $k$ is large enough, there exist codes $C_A$ and $C_B$ of length $\Delta$ such that
    \begin{enumerate}
        \item $\dim C_A = \floor{r\Delta}$ and $\dim C_B = \Delta-\dim C_A$
        \item The distances of $C_A, C_B, C_A^\perp, C_B^\perp$ are all
        at least $\delta \Delta$
        \item Both dual tensor codes $C_0^\perp = (C_A\otimes C_B)^\perp$ and
        $C_1^\perp=(C_A^\perp \otimes C_B^\perp)^\perp$ are $\Delta^{3/2-\eps}$-robust
        with $\Delta^\gamma$-resistance to puncturing
        \item $C_A,C_B,C_A^\perp$, and $C_B^\perp$ have generator matrices where every row and column have at least two ones.
    \end{enumerate}
\end{corollary}
\begin{proof}
We assume for notational simplicity that $r\Delta$ and $(1-r)\Delta$ are integral (the proof is essentially the same without this assumption). We will argue that all four properties are satisfied with probability going to one (as $\Delta$ becomes large) under some distribution for the generation of $C_A$ and $C_B$. By a union bound, a pair satisfying all properties must then exist for large enough $\Delta$.

Consider the distribution over codes $C_A$ and $C_B$ given by generating $C_A$ by a uniformly random $r\Delta \times \Delta$ generator matrix, and $C_B^\perp$ from an independent uniformly random $r\Delta \times \Delta$ generator matrix. Leverrier and Z\'emor \cite{leverrier2022quantum} prove that the first three conditions occur with probability going to one under this distribution (see \cite[Theorem 17]{leverrier2022quantum}), so we need only show the last condition holds.

This follows easily from a few basic observations. Let $r_0 \in (0,1)$ be any constant. First, observe that conditioned on being full rank, a uniformly random $r_0 \Delta \times \Delta$ generator matrix corresponds to a uniformly random subspace of dimension $r_0\Delta$, and furthermore that such a matrix is full rank with probability going to $1$ as $\Delta$ grows large. Second, note that by a Chernoff and union bound, the probability this random generator matrix has any row or column with less than two ones also quickly goes to zero. This implies that for any fixed $r_0$, as $\Delta$ grows large the probability that a random subspace of dimension $r_0\Delta$ has a generator matrix satisfying condition $4$ goes to $1$.

Since $C_A$ and $C_B^\perp$ are generated by uniformly random $r\Delta \times \Delta$ generator matrices, they clearly satisfy condition 4 with high probability. The trick is then simply to notice that (conditioned on full rank), $C_A^\perp$ and $C_B$ are uniformly random subspaces of dimension $(1-r)\Delta$, and therefore also satisfy condition $4$ with probability going to one by the above observation.
\end{proof}
\end{document}